\newif\ifacm \acmtrue
\acmfalse

\ifacm
\documentclass[acmsmall,screen, review, manuscript]{acmart}
\else
\documentclass[a4paper,DIV=15,11pt]{scrartcl}
\fi

\ifacm
\AtBeginDocument{%
  }

\setcopyright{cc}
\copyrightyear{2025}
\acmYear{2025}
\acmDOI{XXXXXXX.XXXXXXX}

\acmConference[ICFP '25]{International Conference on Functional Programming}{October 12--18, 2025}{Singapore}

\acmISBN{978-1-4503-XXXX-X/2025/00}
\fi

\usepackage{amsmath}
\unless\ifacm \usepackage{amssymb,amsthm} \fi
\usepackage{mleftright}

\DeclareEmphSequence{\itshape, \normalfont\sffamily\bfseries, \itshape} 

\usepackage{tikz}
\tikzstyle{bullet}=[circle,fill=black,minimum size=0.1cm,draw,inner sep=0cm]
\usetikzlibrary{trees}
\usetikzlibrary{shapes.geometric}
\usetikzlibrary{positioning}
\usetikzlibrary{tikzmark,calc,topaths}

\unless\ifacm
\usepackage[hidelinks]{hyperref}
\hypersetup{
colorlinks=true,
citecolor=red,
linkcolor=blue,
urlcolor=blue,
}
\fi

\ifacm
\usepackage[nameinlink, noabbrev]{cleveref}
\else
\usepackage[nameinlink, noabbrev, poorman]{cleveref}
\fi

\crefname{section}{Section}{Sections}
\crefname{subsection}{Section}{Sections}
\crefname{appendix}{Appendix}{Appendices}
\crefname{figure}{Figure}{Figures}
\crefname{table}{Table}{Tables}
\crefname{equation}{}{}
\creflabelformat{equation}{#2(#1)#3}

\usepackage{bussproofs, bussproofs-extra}
\newenvironment*{inlineprooftree}{\medskip}{\DisplayProof}%

\ifacm
\theoremstyle{acmplain}%
\newtheorem{theorem}{Theorem}[section]
\crefname{theorem}{Theorem}{Theorems}
\newtheorem{lemma}[theorem]{Lemma}
\crefname{lemma}{Lemma}{Lemmata}
\newtheorem{proposition}[theorem]{Proposition}
\crefname{proposition}{Proposition}{Propositions}
\newtheorem{cor}[theorem]{Corollary}
\crefname{cor}{Corollary}{Corollaries}
\newtheorem{fact}[theorem]{Fact}
\crefname{fact}{Fact}{Facts}

\theoremstyle{acmdefinition}
\newtheorem{definition}[theorem]{Definition}
\crefname{definition}{Definition}{Definitions}
\newtheorem{example}[theorem]{Example}
\crefname{example}{Example}{Examples}

\else
\theoremstyle{plain}%
\newtheorem{theorem}{Theorem}[section]
\newtheorem{lemma}[theorem]{Lemma}
\newtheorem{proposition}[theorem]{Proposition}

\theoremstyle{definition}
\newtheorem{definition}[theorem]{Definition}
\newtheorem{example}[theorem]{Example}

\theoremstyle{remark}

\newtheorem{fact}[theorem]{Fact}
\crefname{fact}{Fact}{Facts}

\crefname{theorem}{Theorem}{Theorems}
\crefname{lemma}{Lemma}{Lemmata}
\crefname{proposition}{Proposition}{Propositions}
\crefname{cor}{Corollary}{Corollaries}

\crefname{definition}{Definition}{Definitions}
\crefname{example}{Example}{Examples}
\crefname{question}{Question}{Questions}
\fi

\newcommand{\proofcase}[1]{\noindent\textbf{Case}{#1}{.}}

\newcommand{\relmiddle}[1]{\mathrel{}\middle#1\mathrel{}}
\newcommand{\mathsetintension}[2]{%
\mleft\{ #1 \relmiddle{|} #2 \mright\}%
}
\newcommand{\mathsetextension}[1]{%
\mleft\{ #1 \mright\}%
}

\newcommand{\mathvect}[1]{\overrightarrow{#1}}

\newcommand{\mathof}[2]{\mathop{#1}\mleft(#2\mright)}

\newcommand{\mathtuple}[1]{\mleft( #1 \mright)}
\newcommand{\mathsequence}[2]{\mleft( #1  \mright)_{#2}}

\newcommand{\mathfunctiondef}[3]{#1\colon #2\to #3}

\renewcommand{\iff}{\mathrel{\Leftrightarrow}}

\newcommand{\mathnat}{\mathbb{N}}

\newcommand{\mathcoloneqq}{\mathrel{::=}}
\newcommand{\mathemptyword}{\varepsilon}

\newcommand{\mathsyof}[2]{\mathop{#1}\mleft(#2\mright)}

\newcommand{\mathsubstbox}[2]{\mathop{#1}\mleft[#2\mright]}
\newcommand{\mathsubst}[2]{#1\mapsto#2}

\newcommand{\mathVar}{\mathrm{Var}}
\newcommand{\mathVarof}[1]{\mathof{\mathrm{\mathVar}}{#1}}

\newcommand{\cmd}[1]{\mathtt{#1}}

\newcommand{\assign}[2]{{#1}\mathrel{:=}{#2}}
\newcommand{\ifelse}[3]{\cmd{if}\;{#1}\;\cmd{then}\;{#2}\;\cmd{else}\;{#3}}
\newcommand{\while}[2]{\cmd{while}\;{#1}\;\cmd{do}\;{#2}\;\cmd{od}}
\newcommand{\config}[2]{\langle {#1}, {#2} \rangle}
\newcommand{\true}{\top}
\newcommand{\false}{\bot}
\newcommand{\sem}[1]{[\![#1]\!]}
\newcommand{\cOR}[2]{\cmd{either}\;{#1}\;\cmd{or}\;{#2}\;\cmd{ro}}

\newcommand{\evaluation}{\mathrel{\longrightarrow}}

\newcommand{\rulename}[1]{(\textsc{#1})}
\newcommand{\mathFVof}[1]{\mathof{\mathrm{FV}}{#1}}

\newcommand{\mathprooffig}[1]{\mathcal{#1}}

\newcommand{\mathquantifier}[1]{\mathcal{#1}}

\newcommand{\triple}[3]{\langle #1 \rangle \,#2\,\langle #3 \rangle}

\newcommand{\IL}{\ensuremath{\mathtt{RHL}}}
\newcommand{\PIL}{\ensuremath{\mathtt{PRHL}}}
\newcommand{\CPIL}{\ensuremath{\mathtt{CPRHL}}}

\newcommand{\htriple}[3]{\mbox{$\{#1\}\,#2\,\{#3\}$}}
\newcommand{\ihtriple}[3]{\mbox{$[#1]\,#2\,[#3]$}}

\newcommand{\tihtriple}[3]{\mleft[\mleft[#1\mright]\,#2\,\mleft[#3\mright]\mright]}
\newcommand{\thtriple}[3]{\mleft\{ \mleft\{ #1 \mright\}\,#2\, \mleft\{ #3 \mright\} \mright\}}

\newcommand{\axiomILrule}{\rulename{Axiom\ensuremath{{}_{\text{\IL}}}}}

\newcommand{\consILrule}{\rulename{Cons\ensuremath{{}_{\text{\IL}}}}}
\newcommand{\seqILrule}{\rulename{Seq\ensuremath{{}_{\text{\IL}}}}}

\newcommand{\assignPILrule}{\rulename{\ensuremath{{:=}_{\text{\PIL}}}}}
\newcommand{\orPILrule}{\rulename{Or\ensuremath{{}_{\text{\PIL}}}}}
\newcommand{\whilePILrule}{\rulename{While\ensuremath{{}_{\text{\PIL}}}}}

\newcommand{\assignCPILrule}{\rulename{\ensuremath{:=_{\text{\CPIL}}}}}
\newcommand{\orCPILrule}{\rulename{Or\ensuremath{{}_{\text{\CPIL}}}}}
\newcommand{\whileCPILrule}{\rulename{While\ensuremath{{}_{\text{\CPIL}}}}}

\newcommand{\mathweakestpresy}{\mathrm{wpr}}
\newcommand{\mathweakestpre}[2]{\mathof{\mathweakestpresy}{#1, #2}}
\newcommand{\mathsetweakestpresy}{\mathbf{WPR}}
\newcommand{\mathsetweakestpre}[2]{\mathof{\mathsetweakestpresy}{#1, #2}}

\newcommand{\mathweakestlibpresy}{\mathrm{wlpr}}
\newcommand{\mathweakestlibpre}[2]{\mathof{\mathweakestlibpresy}{#1, #2}}

\newcommand{\mathstrongestpostsy}{\mathrm{spo}}
\newcommand{\mathstrongestpost}[2]{\mathof{\mathstrongestpostsy}{#1, #2}}

\newcommand{\mathstrongestlibpostsy}{\mathrm{slpo}}
\newcommand{\mathstrongestlibpost}[2]{\mathof{\mathstrongestlibpostsy}{#1, #2}}

\crefname{enumi}{}{}
\crefname{enumii}{}{}

\title{Proof systems for partial incorrectness logic (partial reverse Hoare logic)}
\author{Yukihiro Oda \unless\ifacm \thanks{Tohoku University, yukihiro3socrates6hilbert [at] gmail.com} \fi}

\ifacm
\email{}

\acmArticleType{Research}

\received{???}

\fi

\begin{document}
\ifacm
\begin{abstract}
 \emph{Partial incorrectness logic} (\emph{partial reverse Hoare logic}) has
 recently been introduced as a new Hoare-style logic that over-approximates 
 the weakest pre-conditions of a program and a post-condition.
 It is expected to verify systems  
 where the final state must guarantee its initial state, 
 such as \emph{authentication}, \emph{secure communication tools} and \emph{digital signatures}.
 However, the logic has only been given semantics.
 This paper defines two proof systems 
 for partial incorrectness logic (partial reverse Hoare logic):
 \emph{ordinary} and \emph{cyclic proof systems}.
 They are sound and relatively complete.
 The relative completeness of our ordinary proof system is proved
 by showing that the weakest pre-condition of a while loop and a post-condition is 
 its loop invariant.
 The relative completeness of our cyclic proof system is also proved 
 by providing a way to transform any cyclic proof into an ordinary proof.
\end{abstract}

\ifacm
\fi

\fi

\maketitle

\unless\ifacm
\begin{abstract}
 \emph{Partial incorrectness logic} (\emph{partial reverse Hoare logic}) has
 recently been introduced as a new Hoare-style logic that over-approximates 
 the weakest pre-conditions of a program and a post-condition.
 It is expected to verify systems  
 where the final state must guarantee its initial state, 
 such as \emph{authentication}, \emph{secure communication tools} and \emph{digital signatures}.
 However, the logic has only been given semantics.
 This paper defines two proof systems 
 for partial incorrectness logic (partial reverse Hoare logic):
 \emph{ordinary} and \emph{cyclic proof systems}.
 They are sound and relatively complete.
 The relative completeness of our ordinary proof system is proved
 by showing that the weakest pre-condition of a while loop and a post-condition is 
 its loop invariant.
 The relative completeness of our cyclic proof system is also proved 
 by providing a way to transform any cyclic proof into an ordinary proof.
\end{abstract}

\ifacm
\fi

\fi

\section{Introduction}
\label{sec:intro}
\emph{Hoare-style logics}, such as \emph{partial Hoare logic} \cite{Hoare1969}, 
\emph{total Hoare logic} \cite{Manna1974}, and 
\emph{incorrectness logic}\cite{OHearn2019}, 
also known as \emph{reverse Hoare logic}\cite{Vries2011}, 
are popular logical methods for proving the correctness of programs or for finding bugs,
statically.
They guarantee the corresponding property of a program $C$ 
using triples $\triple{P}{C}{Q}$, 
where $C$ is a program, $P$ is a pre-condition of $C$, and $Q$ is a post-condition of $C$.

Partial Hoare logic \cite{Hoare1969} is the first of all Hoare-style logics and 
guarantees $\htriple{P}{C}{Q}$ its \emph{partial correctness}, i.e. for any state $\sigma$,
if $C$ terminates from $\sigma$ and $P$ holds in $\sigma$, 
then $Q$ holds in the final state.
Partial Hoare logic does not guarantee the termination of a program. 
Total Hoare logic \cite{Manna1974} is an extension of partial Hoare logic
that guarantees termination.

The partial correctness of $\htriple{P}{C}{Q}$ can be restated as follows:
the post-condition $Q$ \emph{over-approximates} 
the \emph{strongest post-condition} of $P$ and $C$, i.e. 
the set of final states in which $C$ terminates from a state in which $P$ holds.
In contrast, incorrectness logic \cite{OHearn2019}, also known as reverse Hoare logic \cite{Vries2011}, guarantees $\tihtriple{P}{C}{Q}$ that
the post-condition $Q$ \emph{under-approximates} the strongest post-condition of $P$ and $C$.
This logic is a method for proving the existence of bugs in $C$ 
rather than proving correctness \cite{OHearn2019,Raad2020}.
This is why the logic is called ``incorrectness'' logic.

L.~Zhang and B.~L.~Kaminski \cite{Zhang2022} defined \emph{partial incorrectness logic}.
It was found by investigating the relation 
between Hoare-style logics and \emph{predicate transformers}, i.e.
\emph{weakest pre-condition}, \emph{weakest liberal pre-condition}, 
strongest post-condition, and \emph{strongest liberal post-condition}.
\cref{fig:triple-precondition-and-postcondition} summarises their results \cite{Zhang2022}.
The logic guarantees $\ihtriple{P}{C}{Q}$ the following:
if $Q$ holds in a state in which $C$ terminates from a state $\sigma$,
then $P$ holds in $\sigma$.

The logic was named ``partial incorrectness logic'', perhaps because of
the relation with total Hoare logic.
It is like the relation between partial Hoare logic and 
incorrectness logic (reverse Hoare logic).
Actually, incorrectness logic (reverse Hoare logic) is a total logic 
because it requires the termination of a target program \cite{Vries2011}.
Because the logic is not so, it is a partial logic.
However, as we will see later, the logic is to prove ``correctness'',
not to find bugs.
Therefore, in this paper, we mainly refer to it as \emph{partial reverse Hoare logic}.

\begin{table*}[t]
 \centering
 \begin{tabular}[t]{|c|c|c|}
  \hline
  Logic & Pre-condition &  Post-condition \\
  \hline
  \hline
  Total Hoare logic $\thtriple{P}{C}{Q}$ & $P \Rightarrow \mathweakestpre{C}{Q} $ &  None \\
  Partial Hoare logic $\htriple{P}{C}{Q}$ & $P \Rightarrow \mathweakestlibpre{C}{Q} $ & $\mathstrongestpost{P}{C} \Rightarrow Q$   \\
   Partial incorrectness logic (Partial reverse Hoare logic) $\ihtriple{P}{C}{Q}$ & $\mathweakestpre{C}{Q} \Rightarrow  P$ &  $Q \Rightarrow \mathstrongestlibpost{P}{C}$ \\
  (Total) incorrectness logic  (reverse Hoare logic)  $\tihtriple{P}{C}{Q}$ & None  & $Q \Rightarrow \mathstrongestpost{P}{C}$  \\
  \hline
 \end{tabular}

 \begin{align*}
 \sigma \models \mathweakestpre{C}{Q} 
  &\iff \exists \sigma' \mleft( \config{C}{\sigma} \to  \config{\mathemptyword}{\sigma'} \land {\sigma' \models Q} \mright) &\text{(Weakest pre-condition)} \\
 \sigma \models \mathweakestlibpre{C}{Q} 
  &\iff \forall \sigma' \mleft( \config{C}{\sigma} \to  \config{\mathemptyword}{\sigma'} \Rightarrow {\sigma' \models Q} \mright) &\text{(Weakest liberal pre-condition)} \\
 \sigma' \models \mathstrongestpost{P}{C}
  &\iff \exists \sigma \mleft( \config{C}{\sigma} \to  \config{\mathemptyword}{\sigma'} \land {\sigma \models P} \mright) &\text{(Strongest post-condition)} \\
  \sigma' \models \mathstrongestlibpost{P}{C}
  &\iff \forall \sigma \mleft( \config{C}{\sigma} \to  \config{\mathemptyword}{\sigma'} \Rightarrow {\sigma \models P} \mright) &\text{(Strongest liberal post-condition)}
 \end{align*}

 \caption{Triples and Pre-condition/post-condition (Similar to the table on \cite[p.87:20]{Zhang2022})}
 \label{fig:triple-precondition-and-postcondition} 
\end{table*}

To see the usefulness of partial reverse Hoare logic,
we describe the case of a password-based authentication system 
given by L.~Zhang and B.~L.~Kaminski \cite{Zhang2022}.
Consider a program $C_{\text{Auth}}$ for a password-based authentication system that
takes as input ``username'', ``password'', and so on, and then outputs ``approved'' 
if the user is identified and ``rejected'' otherwise.
When we check this program using Hoare-style logic, we naively think the following triple:
\begin{equation} 
 \triple{\cmd{username}= x, \cmd{password} = y, \dots \text{ are correct}}{C_{\text{Auth}}}{\text{``approved''}}. \label{eq:example}
\end{equation}

If \cref{eq:example} were proved in partial Hoare logic, 
then the program would be guaranteed the following: 
the program outputs ``approved'' if the inputs are correct.
But, this is wrong: it does not guarantee that the wrong user will be rejected.
Someone might think that we just need to show the following triple:
\begin{equation}
  \triple{\cmd{username}= x, \cmd{password} = y, \dots \text{ are not correct}}{C_{\text{Auth}}}{\text{``rejected''}}. \label{eq:exampleneg}
\end{equation}
But this means that we have to show two triples, which takes time and effort.
Moreover, this is over-engineering: 
this analysis consequently guarantees that the program outputs ``approved'' 
if and only if the inputs are correct, even if unexpected errors occur.
This problem occurs in total Hoare logic.

If \cref{eq:example} were proved in incorrectness logic (reverse Hoare logic),
then the program would be guaranteed that 
the inputs could be correct if the program outputs ``approved''.
It guarantees nothing. 
Someone might think that, considering \cref{eq:exampleneg}, all is well.
But they are wrong.
If \cref{eq:exampleneg} were proved in reverse Hoare logic,
then it would only guarantee that
the inputs could not be correct if the program outputs ``rejected''.
This is not what we want to guarantee.
From the view of ``incorrectness'', we may have to prove the following triple,
which means that there are some bugs in $C_{\text{Auth}}$:
\[
   \triple{\cmd{username}= x, \cmd{password} = y, \dots \text{ are not correct}}{C_{\text{Auth}}}{\text{``approved''}}. \label{eq:example-neg}
\]
On the other hand, this means that in this case, incorrectness logic can find bugs
but cannot guarantee ``correctness''.

Now, suppose that \cref{eq:example} is proved in partial reverse Hoare logic.
In this case, the inputs must be correct if the program outputs ``approved''.
\emph{This is what we want to guarantee for $C_{\text{Auth}}$}.
It also solved the problem of over-engineering problem in partial Hoare logic.
This result allows the situation where the program outputs ``rejected''
if the inputs are correct, but an unexpected error occurs.
In this example, ``incorrectness'' is not essential. 
When we prove \cref{eq:example},
we do not find any bugs in $C_{\text{Auth}}$, but
we show the ``reverse-correctness'' of $C_{\text{Auth}}$.
Then, in our opinion, the name ``partial incorrectness logic'' is not appropriate.
We would like to call the logic ``partial reverse Hoare logic''.

As we see in the above case, partial reverse Hoare logic is useful for verifying systems 
in which the final state must guarantee its initial state, 
such as authentication, secure communication tools, and digital signatures.
L.~Verscht, Ā.~Wáng and B.~L.~Kaminski \cite{Verscht2025P} also give some useful cases of
partial reverse Hoare logic.

\subsection{Our contribution}
Our main contribution is to define two proof systems 
for partial incorrectness logic (partial reverse Hoare logic):
\emph{ordinary} and \emph{cyclic proof systems}.
These systems, which are sound and relatively complete, 
will have practical applications in
software verification for secure systems.
While L.\ Zhang and B.\ L.\ Kaminski \cite{Zhang2022} defined 
the semantics of partial incorrectness logic, they did not give its proof system.
Therefore, our systems are the first for partial reverse Hoare logic, 
opening up new possibilities for practical use and further research in the field.

In our ordinary proof system, where every proof figure is a finite tree,
the rule for the while loop is the dual of the corresponding rule 
in ``partial'' Hoare logic, not in ``total'' Hoare logic.
We note that the semantics of partial reverse Hoare logic is the dual of total Hoare logic,
so there is a twist.
This twist is very interesting, but we do not know why.

\emph{Cyclic proof systems} are proof systems that allow cycles in proof figures.
When the rule for the while loop is applied in our ordinary proof system,
we have to find a good loop invariant, just as in partial Hoare logic, 
which is challenged \cite{Furia2014}.
In contrast, we do not have to find loop invariants
when the while loop rule is applied in our cyclic proof system.
Hence, cyclic proofs have an advantage in proof search.

We give outlines for proofs of the relative completeness.
We show the relative completeness of our ordinary proof system
by showing that the weakest pre-condition predicate of 
a while loop and a post-condition is its loop invariant.
We also prove that of our cyclic proof system 
by giving a way to transform any cyclic proof into an ordinary proof.

\subsection{Related work}
We present related work.

The results in Hoare logic are too numerous to be presented here. 
However,  
there is a detailed survey of Hoare logic by K.~R.~Apt and E.~Olderog \cite{Apt2O19}. 
One of the most important recent extensions of Hoare logic is \emph{separation logic}, 
which is used to reason about pointer structures \cite{OHearn2001,Reynolds2002}.
It is applied in \emph{Infer}\cite{Calcagno2011}, used by Meta, 
\emph{Prusti}\cite{Astrauskas2022}, a verifier for Rust, and
\emph{Iris}\cite{Jung2015}, implemented and verified in Coq, and so on.

Incorrectness logic \cite{OHearn2019}, reverse Hoare logic \cite{Vries2011},
is used mainly to find bugs.
It has been extended by separation logic 
to \emph{incorrectness separation logic} \cite{Raad2020} and
\emph{concurrent incorrectness separation logic} \cite{Raad2O22}.
Y. Lee and K. Nakazawa \cite{Lee2024} showed 
the relative completeness of incorrectness separation logic.

As mentioned above,
partial incorrectness logic (partial reverse Hoare logic) was found 
by investigating the relation 
between Hoare-style logics and predicate transformers \cite{Zhang2022}.
L.~Verscht and B.~L.~Kaminski \cite{Verscht2025T} investigated 
this relation in more detail.
L.~Verscht, Ā.~Wáng and B.~L.~Kaminski \cite{Verscht2025P} introduced
partial incorrectness logic from the point of view of predicate transformers
and discussed some useful cases.

Non-well-founded proof systems are a type of proof system 
that allows a proof figure to contain infinite paths.
Cyclic, or circular, proof systems are a type of non-well-founded proof system.
It allows a proof figure to contain cycles.
R.~N.~S.~Rowe \cite{Rowes-list} summarises the extensive list of
academic work on cyclic and non-well-founded proof theory.
Cyclic proofs are defined for some logics or theories 
to reason about inductive or recursive structures, 
such as modal $\mu$-calculus \cite{Afshari2017}, 
G\"{o}del-L\"{o}b provability logic \cite{Shamkanov2014}, 
first-order logic with inductive definitions \cite{BrotherstonS2011,Berardi2019,Oda2024}, 
arithmetic\cite{Simpson2017, Das2019}, separation logic \cite{Brotherston2008,BrotherstonA2011,Tatsuta2019,Kimura2020,Saotome2020,Saotome2024}.
Cyclic proofs are also useful for software verification
because of their finiteness,
for example, abduction \cite{Brotherston2014}, 
termination of pointer programs \cite{Brotherston2008,Rowe2017}, 
temporal property verification\cite{Tellez2020},
solving horn clauses \cite{Unno2017}, 
model checking\cite{Tsukada2022}, and
decision procedures for symbolic heaps \cite{Brotherston2012, Chu2015, Ta2016, Ta2018, Tatsuta2019}.

\subsection{Outline of this paper}
We outline the rest of this paper.
\cref{sec:syntax_semantics} introduces 
the syntax and operational semantics of our non-deterministic target language.
In \cref{sec:partial-reverse-Hoare-logic},
we define an ordinary proof system for partial reverse Hoare logic and
show its soundness and relative completeness.
\cref{sec:cyclic-proofs-PRHL} gives cyclic proofs and 
shows that their provability is the same as that of our ordinary proof system.
\cref{sec:conc} concludes.
\section{Programs and assertions}
\label{sec:syntax_semantics}
This section introduces the syntax and semantics of our  non-deterministic target languages.

Let $\mathnat$ be the whole of natural numbers, that is $\mathsetextension{0, 1, \dots}$, 
and $\mathVar$ be an infinite set of \emph{variables}. \emph{Expressions} $E$, \emph{Boolean conditions} $B$, \emph{programs} $C$, and \emph{Assertions} $P$ are defined as the following grammar:
\begin{align*}
 E &\mathcoloneqq  {
 {x} \mid {n} \mid \mathsyof{f}{E, \dots, E} 
 }, \\
 B &\mathcoloneqq  { 
 \mathsyof{Q}{E, \dots, E} \mid {E = E} \mid 
 {E \leq E}  \mid {\lnot B}  \mid {B \land B} \mid {B \lor B}
 },  \\
 C &\mathcoloneqq {\mathemptyword} \mid {C'} \\
 C' &\mathcoloneqq  {\assign{x}{E}} \mid {C';C'} \mid {\while{B}{C}} \mid  {\cOR{C}{C}}, \\
 {P} &\mathcoloneqq {{B} \mid {\lnot P} \mid {P \lor P} \mid {P \land P} \mid {P \to P}  \mid {\exists x(P)} \mid {\forall x(P)} } 
\end{align*}
where $\mathemptyword$ denotes the empty string, and
$x$, $n$, $f$ and $Q$ range over 
$\mathVar$, $\mathnat$, the set of functions $\mathnat \to \mathnat$ and the set of predicates or relations on $\mathnat$, respectively.

For simplicity, 
we restrict control flow statements to only ${\while{B}{C}}$ and ${\cOR{C_{0}}{C_{1}}}$.
However, many control flow statements can be simulated in our language.
For example, $\ifelse{B}{C_{0}}{C_{1}}$ can be simulated by 
$\assign{x}{0};\while{B \land x=0}{C_{0};\assign{x}{1}};\while{\lnot B \land x=0}{C_{1};\assign{x}{1}}$ 
with some fresh variable $x$.

In each occurrence of the form ${\exists x(P)}$ and ${\forall x(P)}$, we say that 
the occurrence $x$ is \emph{binding} with \emph{scope $P$}.
We say that an occurrence of a variable is \emph{bound} 
if it is a binding occurrence of the variable.
An occurrence of a variable is called \emph{free} if it is not bound.
As usual, 
we assume that $\alpha$-conversions (renaming of bound variables) are implicitly applied
in order that bound variables are always different from each other
and from free variables. 
We write $\mathFVof{P}$ for the set of free variables occurring in an assertion $P$.
We write $\mathVarof{E}$, $\mathVarof{B}$, and $\mathVarof{C}$, for the set of variables 
occurring in expression $E$, Boolean condition $B$, and program $C$.
As usual, we write $E_{0} \neq E_{1}$ for $\lnot(E_{0} = E_{1})$ and
$E_{0} > E_{1}$ for $\lnot(E_{0} \leq E_{1})$.
We write $\mathsubstbox{E}{\mathsubst{x}{E'}}$ and $\mathsubstbox{B}{\mathsubst{x}{E'}}$ to denote the substitution of expression $E'$ for variable $x$ in expression $E$ and Boolean condition $B$, respectively.
We write ${\mathsubstbox{P}{\mathsubst{x_{0}}{E_{0}}, \dots, \mathsubst{x_{n}}{E_{n}}}}$
for the assertion obtained 
by replacing all the free occurrences of $x_{0}, \ldots, x_{n}$ in $P$
with $E_{0}, \ldots, E_{n}$.
For ${\mathquantifier{Q}_{1}, \dots, \mathquantifier{Q}_{n}}\in{\mathsetextension{\exists, \forall}}$,
we abbreviate $\mathquantifier{Q}_{1} x_{1} (\dots (\mathquantifier{Q}_{n} x_{n} (P)) \dots)$ to $\mathquantifier{Q}_{1} x_{1} \dots \mathquantifier{Q}_{n} x_{n} (P)$.
For programs $C_{0}$ and $C_{1}$,
we write $C_{0};C_{1}$ for $C_{i}$ if ${C_{1-i}}\equiv{\mathemptyword}$ and
$C_{0};C_{1}$ otherwise.

A \emph{(program) state} is defined as a function $\mathfunctiondef{\sigma}{\mathVar}{\mathnat}$.  
We define the semantics $\sem{E}\sigma \in \mathnat$ and $\sem{B}\sigma \in \{\true,\false\}$ 
of expression $E$ and Boolean condition $B$ in state $\sigma$ in the usual way: 
\begin{align*}
 \sem{n}\sigma &= n \quad \text{ for a natural number } n, \\
 \sem{x}\sigma &= \mathof{\sigma}{x} \quad \text{ for } x\in\mathVar, \\
 \sem{\mathsyof{f}{E_{0}, \dots, E_{n}}}\sigma &= \mathsyof{f}{\sem{E_{0}}\sigma, \dots, \sem{E_{n}}\sigma}, \\
 \sem{\mathsyof{Q}{E_{0}, \dots, E_{n}}}\sigma=\true &\iff  \mathtuple{\sem{E_{0}}\sigma, \dots, \sem{E_{n}}\sigma}\in Q, \\
 \sem{E = E'}\sigma=\true &\iff {\sem{E}\sigma}={\sem{E'}\sigma}, \\
 \sem{E \leq E'}\sigma=\true &\iff \sem{E}\sigma\leq\sem{E'}\sigma,  \\
 \sem{\lnot B}\sigma=\true &\iff {\sem{B}\sigma}=\false, \\
 \sem{B \land B'}\sigma=\true &\iff {\sem{B}\sigma}=\true \text{ and } {\sem{B'}\sigma}=\true, \\
 \sem{B \lor B'}\sigma=\true &\iff {\sem{B}\sigma}=\true \text{ or } {\sem{B'}\sigma}=\true.
\end{align*}
We write $\mathsubstbox{\sigma}{\mathsubst{x}{E}}$ for the state defined 
as $\sigma$ on all variables except $x$, with $\mathof{\mathsubstbox{\sigma}{\mathsubst{x}{E}}}{x} = \sem{E}\sigma$.  

\begin{lemma}
\label[lemma]{lem:expr_subst}
For all expressions $E$ and $E'$, Boolean expressions $B$, program states $\sigma$ and variables $x$, the following statements hold:
 \begin{align*}
  \sem{\mathsubstbox{E}{\mathsubst{x}{E'}}}\sigma &=  \sem{E}(\mathsubstbox{\sigma}{\mathsubst{x}{\sem{E'}\sigma}}), \text{ and } \\
  \sem{\mathsubstbox{B}{\mathsubst{x}{E}}}\sigma &\iff  \sem{B}(\mathsubstbox{\sigma}{\mathsubst{x}{\sem{E}\sigma}}). 
 \end{align*}
\end{lemma}

\begin{proof}
 By structural induction on $E$ and $B$, respectively.
\end{proof}

Satisfaction of an assertion $P$ by a state $\sigma$, written $\sigma \models P$, is defined inductively as follows:
\begin{align*}
 \sigma \models B & \iff  \sem{B}\sigma = \top, \\
 \sigma \models  \lnot P_{0} & \iff  \sigma \not\models P_{0}, \\
 \sigma \models  P_{0} \lor P_{1} & \iff  \sigma \models P_{0} \text{ or } \sigma \models P_{1}, \\
 \sigma \models  P_{0} \land P_{1} & \iff  \sigma \models P_{0} \text{ and } \sigma \models P_{1}, \\
 \sigma \models  P_{0} \to P_{1} & \iff  \sigma \not\models P_{0} \text{ or } \sigma \models P_{1}, \\
 \sigma \models  \exists x(P_{0}) & \iff \mathsubstbox{\sigma}{\mathsubst{x}{c}} \models P_{0} \text{ for some } c, \\
 \sigma \models  \forall x(P_{0}) & \iff  \mathsubstbox{\sigma}{\mathsubst{x}{c}} \models P_{0} \text{ for all } c.
\end{align*}

For assertions $P$ and $Q$, we write ${P}\models{Q}$ 
if ${\sigma}\models{P}$ implies ${\sigma}\models{Q}$ for any state $\sigma$.
For an assertion $P$, we write $\models{P}$ 
if ${\sigma}\models{P}$ holds for any state $\sigma$.
\begin{lemma}[Substitution]
\label[lemma]{lem:assert_subst}
For all assertions $P$, program states $\sigma$, expressions $E$ and variables $x$,
\[
\sigma \models \mathsubstbox{P}{\mathsubst{x}{E}} \text{ if and only if } \sigma[x \mapsto \sem{E}\sigma] \models P.
\]
\end{lemma}

\begin{proof}[Proof sketch]
 The `if' part: By structural induction on $P$. 

 The `only if' part: By structural induction on $\mathsubstbox{P}{\mathsubst{x}{E}}$. 
\end{proof}

\begin{figure*}[t]
 \centering
 \begin{align*}
  \config{\assign{x}{E}}{\sigma} &\evaluation  \config{\mathemptyword}{\mathsubstbox{\sigma}{\mathsubst{x}{\sem{E}\sigma}}} &&  &(\cmd{assign}) \\
  \config{\while{B}{C}}{\sigma} &\evaluation \config{C;\while{B}{C}}{\sigma}  &&\text{if } \sem{B}\sigma = \true &\quad (\cmd{while}\ 1) \\
  \config{\while{B}{C}}{\sigma} &\evaluation \config{\mathemptyword}{\sigma}  &&\text{if } \sem{B}\sigma = \false &(\cmd{while}\ 2) \\
  \config{C_{0};C_{1}}{\sigma} &\evaluation \config{C'_{0};C_{1}}{\sigma'} &&\text{if } \config{C_{0}}{\sigma} \evaluation \config{C'_{0}}{\sigma'}  &(\cmd{seq}) \\
\config{\cOR{C_{0}}{C_{1}}}{\sigma} &\evaluation \config{C_{i}}{\sigma} &&\text{for } i=0, 1  &(\cmd{or}\ i) 
 \end{align*}
 \caption{Small-step semantics of programs} \label{fig:prog_sem}
\end{figure*}

A \emph{(program) configuration} is defined as a pair $\config{C}{\sigma}$,
where $C$ and $\sigma$ are a program and a state, respectively.
In \cref{fig:prog_sem}, we define the operational semantics of our programs 
by giving the small-step relation $\evaluation$ on configurations.  
An \emph{execution} (of $C$) is defined as a possibly infinite sequence of configurations $\mathsequence{\config{C_{i}}{\sigma_{i}}}{i \geq 0}$ with $C_{0} = C$ such that $\config{C_{i}}{\sigma_{i}} \evaluation \config{C_{i}}{\sigma_{i+1}}$ for all $i \geq 0$.
For a finite execution $\mathsequence{\config{C_{i}}{\sigma_{i}}}{0 \leq i \leq n}$,
the \emph{length of the finite execution} is defined as $n$.
 We write $\evaluation^{n}$ for an $n$-step execution.
We also sometimes write $\evaluation^{*}$ 
for the reflexive-transitive closure of $\evaluation$. 


\begin{lemma}
 \label[lemma]{lem:property-config}
 The following statements hold:
 \begin{enumerate}
  \item For a program $C$, states $\sigma$, $\sigma'$,
	and a variable ${z}\notin{\mathVarof{C}}$,
	if ${\config{{C}}{\sigma}} \evaluation^{*} {\config{\mathemptyword}{\sigma'}}$ holds,
	then ${\mathof{\sigma}{z}}={\mathof{\sigma'}{z}}$ holds.
	\label{item:not-in-lem-property-config}
  \item For a program $C$, states $\sigma$, $\sigma'$,
	and a variable ${z}\notin{\mathVarof{C}}$,
	if ${\config{{C}}{\mathsubstbox{\sigma}{\mathsubst{z}{a}}}} \evaluation^{*} {\config{\mathemptyword}{\mathsubstbox{\sigma'}{\mathsubst{z}{a}}}}$ holds,
	then ${\config{{C}}{\mathsubstbox{\sigma}{\mathsubst{z}{\sem{z}\sigma'}}}} \evaluation^{*} {\config{\mathemptyword}{\sigma'}}$ holds.
	\label{item:not-in-lem-property-config2}
  \item For programs $C_{0}$ and $C_{1}$, and states $\sigma$ and $\sigma'$,
	 ${\config{C_{0};C_{1}}{\sigma}} \evaluation^{*} {\config{\mathemptyword}{\sigma'}}$ 
	holds if and only if there exists $\sigma''$ such that
	${\config{C_{0}}{\sigma}} \evaluation^{*} {\config{\mathemptyword}{\sigma''}}$ and
	${\config{C_{1}}{\sigma''}} \evaluation^{*} {\config{\mathemptyword}{\sigma'}}$ hold.
	\label{item:sequencing-lem-property-config}
  \item Let $B$ be a Boolean condition, and $C$ be a program. 
	For states $\sigma$ and $\sigma'$,
	${\config{\while{B}{C}}{\sigma}} \evaluation^{*} {\config{\mathemptyword}{\sigma'}}$ holds
	if and only if
	there exist states $\sigma_{0}, \dots, \sigma_{k}$ 
	such that ${\sigma_{0}}\equiv{\sigma}$, ${\sigma_{k}}\equiv{\sigma'}$ and
	${\sigma_{k}}\models{\lnot B}$ hold, and
	$k>0$ implies that
	${\config{C}{\sigma_{i}}} \evaluation^{*} {\config{\mathemptyword}{\sigma_{i+1}}}$ and
	${\sigma_{i}}\models{B}$
	for each $i=0, \dots, k-1$.
	\label{item:while-lem-property-config}
 \end{enumerate}
\end{lemma}

\begin{proof}[Proof(Sketch)] We  give the outline of proof of each statement.

 \noindent \cref{item:not-in-lem-property-config}
 By induction on the length of ${\config{{C}}{\sigma}} \evaluation^{*} {\config{\mathemptyword}{\sigma'}}$.

 \noindent \cref{item:not-in-lem-property-config2}
 By induction on the length of ${\config{{C}}{\mathsubstbox{\sigma}{\mathsubst{z}{a}}}} \evaluation^{*} {\config{\mathemptyword}{\mathsubstbox{\sigma'}{\mathsubst{z}{a}}}}$.


 \noindent \cref{item:sequencing-lem-property-config}
 By induction on the length of 
 ${\config{C_{0};C_{1}}{\sigma}} \evaluation^{*} {\config{\mathemptyword}{\sigma'}}$.

 \noindent \cref{item:while-lem-property-config}
 \noindent The `if' part: By induction on $k$.

 \noindent The `only if' part: 
 By induction on the length of
 ${\config{\while{B}{C}}{\sigma}} \evaluation^{*} {\config{\mathemptyword}{\sigma'}}$.

\end{proof}
\section{An ordinary proof system for partial incorrectness logic (partial reverse Hoare logic)}
\label{sec:partial-reverse-Hoare-logic}
This section introduces \emph{partial incorrectness logic} 
(\emph{partial reverse Hoare logic}) and its ordinary proof system.
Our proof system is similar to partial Hoare logic, 
except for the composition rule and the rule for $\while{B}{C}$.
Interestingly, 
although the semantics of partial reverse Hoare logic is the dual of ``total'' Hoare logic,
the rule for $\while{B}{C}$ is the dual of the corresponding rule 
in ``partial'' Hoare logic, not in total Hoare logic.

We write \emph{partial reverse Hoare triples} as $\ihtriple{P}{C}{Q}$, 
where $C$ is a program, and $P$ and $Q$ are assertions. 
Partial reverse Hoare triples are the same as \emph{partial incorrectness triples}
\cite{Zhang2022,Verscht2025T}.
As we said in \cref{sec:intro}, 
``incorrectness'' is not essential for partial reverse Hoare logic.
That is why we do not use the term ``partial incorrectness triples''.

\begin{definition}
\label[definition]{def:reverse-correctness}
A partial reverse Hoare triple $\ihtriple{P}{C}{Q}$ is said to be \emph{valid} if, 
for all states $\sigma'$ with $\sigma' \models Q$,
the following state holds:
for any state $\sigma$,
if $\config{C}{\sigma} \evaluation^{*} \config{\mathemptyword}{\sigma'}$ holds,
then ${\sigma}\models{P}$.
\end{definition}
This definition is equivalent to Definition 6.4 in \cite{Zhang2022}, 
which is shown in \cite[p.87:20]{Zhang2022}. 
To see the equivalence, 
we describe the relationship between the validity of partial reverse Hoare triples
and the weakest pre-condition.
\begin{definition}[Weakest pre-condition]
 \label[definition]{definition:weakest-pre-condition}
 For an assertion $Q$ and a program $C$,
 we define a set of states $\mathsetweakestpre{C}{Q}$ by:
 \[
 \mathsetweakestpre{C}{Q}=\mathsetintension{\sigma}{
 \begin{aligned}
 &\text{There exists a state } \sigma' \text{ such that } \\
 &{\config{C}{\sigma} \evaluation^{*} \config{\mathemptyword}{\sigma'}} \text{ and } 
 {{\sigma'} \models {Q}} \text{ hold} 
 \end{aligned}
 }.
 \]
\end{definition}

Intuitively, $\sigma\in\mathsetweakestpre{C}{Q}$ holds if and only if
$C$ terminates from $\sigma$ and the final state satisfies $Q$.
Then, we see the following statement, 
which means that \cref{def:reverse-correctness} is equivalent to 
Definition 6.4 in \cite{Zhang2022}.
\begin{proposition}
 \label[proposition]{prop:weakestness-set}
 $\ihtriple{P}{C}{Q}$ is valid
 if and only if 
 ${\mathsetweakestpre{C}{Q}}\subseteq{\mathsetintension{\sigma}{{\sigma}\models{P}}}$ 
 holds.
\end{proposition}

\begin{proof}
 \noindent The `if' part:
 Assume 
 ${\mathsetweakestpre{C}{Q}}\subseteq{\mathsetintension{\sigma}{{\sigma}\models{P}}}$.
 Fix a state $\sigma'$ with ${\sigma'}\models{Q}$.
 Fix a state $\sigma$ with $\config{C}{\sigma} \evaluation^{*} \config{\mathemptyword}{\sigma'}$.
 Then, ${\sigma}\in{\mathsetweakestpre{C}{Q}}$ holds.
 Because of 
 ${\mathsetweakestpre{C}{Q}}\subseteq{\mathsetintension{\sigma}{{\sigma}\models{P}}}$,
 we have ${\sigma}\models{P}$.
 Thus, $\ihtriple{P}{C}{Q}$ is valid.

 \noindent The `only if' part:
 Assume that $\ihtriple{P}{C}{Q}$ is valid.
 Let ${\sigma}\in{\mathsetweakestpre{C}{Q}}$.
 Then, there exists a state $\sigma'$ such that
 ${\config{C}{\sigma} \evaluation^{*} \config{\mathemptyword}{\sigma'}}$ and
 ${{\sigma'} \models {Q}}$ hold.
 Since $\ihtriple{P}{C}{Q}$ is valid,
 ${\sigma}\models{P}$.
 Thus, ${\mathsetweakestpre{C}{Q}}\subseteq{\mathsetintension{\sigma}{{\sigma}\models{P}}}$
 holds.
\end{proof}

\begin{figure*}[t]
 \centering
 \begin{inlineprooftree}
  \AxiomC{}
  \RightLabel{\axiomILrule}
  \UnaryInfC{$\ihtriple{Q}{\mathemptyword}{Q}$}
 \end{inlineprooftree}
 \begin{inlineprooftree}
  \AxiomC{\phantom{$\ihtriple{\mathsubstbox{P}{\mathsubst{x}{E}}}{\assign{x}{E}}{P}$}}
  \RightLabel{\assignPILrule}
  \UnaryInfC{$\ihtriple{\mathsubstbox{Q}{\mathsubst{x}{E}}}{\assign{x}{E}}{Q}$}
 \end{inlineprooftree}
 \begin{inlineprooftree}
  \AxiomC{$\ihtriple{P}{C_{0}}{R}$}
  \AxiomC{$\ihtriple{R}{C_{1}}{Q}$}
  \RightLabel{\seqILrule}
  \BinaryInfC{$\ihtriple{P}{{C_{0}};{C_{1}}}{Q}$}
 \end{inlineprooftree}
 \begin{inlineprooftree}
  \AxiomC{$\ihtriple{P}{C}{Q}$}
  \LeftLabel{(${P}\models{P'}$, ${Q'}\models{Q}$)}
  \RightLabel{\consILrule}
  \UnaryInfC{$\ihtriple{P'}{C}{Q'}$}
 \end{inlineprooftree}
 \begin{inlineprooftree}
  \AxiomC{$\ihtriple{P}{C_{0}}{Q}$}
  \AxiomC{$\ihtriple{P}{C_{1}}{Q}$}
  \RightLabel{\orPILrule}
  \BinaryInfC{$\ihtriple{P}{\cOR{C_{0}}{C_{1}}}{Q}$}
 \end{inlineprooftree}
 \begin{inlineprooftree}
  \AxiomC{$\ihtriple{{B}\to{P}}{C}{P}$}
  \RightLabel{\whilePILrule}
  \UnaryInfC{$\ihtriple{P}{\while{B}{C}}{{\lnot B}\to{P}}$}
 \end{inlineprooftree}
 
 \caption{The rules for our ordinary proof system of partial incorrectness logic (partial reverse Hoare logic)}
 \label{fig:rules-for-ordinary-proof-of-par-incor}
\end{figure*}

\cref{fig:rules-for-ordinary-proof-of-par-incor} shows the inference rules for our proof system of \emph{partial reverse Hoare logic}.
We note two points about these rules.

Firstly, our rule for assignment is similar to that in Hoare logic, 
not in (total) reverse Hoare logic.
The naive translation of \assignPILrule
\begin{prooftree}
 \AxiomC{\phantom{$\ihtriple{\mathsubstbox{P}{\mathsubst{x}{E}}}{\assign{x}{E}}{P}$}}
 \UnaryInfC{$\tihtriple{\mathsubstbox{Q}{\mathsubst{x}{E}}}{\assign{x}{E}}{Q}$}
\end{prooftree}
is not sound and complete in (total) reverse Hoare logic 
(see \cite[p.159]{Vries2011} for details).
However, \assignPILrule\ is sound and complete in partial reverse Hoare logic,
as we show later.

Secondly, assertions in our rule for the while loop are 
the dual of these in partial Hoare logic.
The popular rule for the while loop in partial Hoare logic is as follows:
\begin{center} 
 \begin{inlineprooftree}
  \AxiomC{$\htriple{{B}\land{P}}{C}{P}$}
  \UnaryInfC{$\htriple{P}{\while{B}{C}}{{\lnot B}\land{P}}$}
 \end{inlineprooftree}.
\end{center}
We note that ${P\land B}\models{P'}$ is equivalent to ${P}\models{B\to P'}$.
Interestingly, \whilePILrule\ is sound and complete in partial reverse Hoare logic.
As we see in \cref{sec:intro}, 
the semantics of partial reverse Hoare logic is the dual of ``total'' Hoare logic.
However, the rule is not so; it is the dual of partial Hoare logic.
This fact is very interesting, but we do not understand why the twist exists.

We call $P$ in \whilePILrule\ a \emph{loop invariant of $\while{B}{C}$}.

We define a \emph{\PIL -proof} as a derivation tree constructed 
according to the proof rules in \cref{fig:rules-for-ordinary-proof-of-par-incor},
each leaf of which is a conclusion of either \axiomILrule\ or \assignPILrule.
If there is a \PIL -proof whose root is labelled by $\ihtriple{P}{C}{Q}$,
we say that $\ihtriple{P}{C}{Q}$ is \emph{provable in \PIL}.

\begin{example}
 Let $+$ addition operator. The following is a \PIL -proof:
 \begin{center}
  \begin{inlineprooftree}
   \AxiomC{}
   \RightLabel{\assignPILrule}
   \UnaryInfC{$\ihtriple{\top}{\assign{x}{x+i}}{\top}$}
   \RightLabel{\consILrule}
   \UnaryInfC{$\ihtriple{{i<5}\to{\top}}{\assign{x}{x+i}}{\top}$}

   \AxiomC{}
   \RightLabel{\assignPILrule}
   \UnaryInfC{$\ihtriple{\top}{\assign{i}{i+1}}{\top}$}
   \RightLabel{\consILrule}
   \UnaryInfC{$\ihtriple{{i<5}\to{\top}}{\assign{i}{i+1}}{\top}$}

   \RightLabel{\seqILrule}
   \BinaryInfC{$\ihtriple{{i<5}\to{\top}}{{\assign{x}{x+i}};{\assign{i}{i+1}}}{\top}$}
   \RightLabel{\whilePILrule}
   \UnaryInfC{$\ihtriple{\top}{\while{i<5}{{\assign{x}{x+i}};{\assign{i}{i+1}}}}{{\lnot (i < 5)}\to \top}$}
   \RightLabel{\consILrule}
   \UnaryInfC{$\ihtriple{\top}{\while{i<5}{{\assign{x}{x+i}};{\assign{i}{i+1}}}}{{x>0}\land{i\geq 5}}$}
  \end{inlineprooftree}.
 \end{center}
\end{example}

We show the soundness theorem.
\begin{proposition}[Soundness]
\label{prop:rev_hoare_sound}
If $\ihtriple{P}{C}{Q}$ is provable in \PIL, then it is valid.
\end{proposition}

\begin{proof}
 It suffices to show the local soundness of each rule: 
 if all the premises are valid, then the conclusion is valid.

 \proofcase{\axiomILrule} Obvious.

 \proofcase{\assignPILrule}
 We show that $\ihtriple{\mathsubstbox{Q}{\mathsubst{x}{E}}}{\assign{x}{E}}{Q}$ is valid.
 
 Fix a state $\sigma'$ with ${\sigma'} \models {Q}$.
 Fix a state $\sigma$ with $\config{\assign{x}{E}}{\sigma} \evaluation^{*} \config{\mathemptyword}{\sigma'}$.
 Since $\config{\assign{x}{E}}{\sigma} \evaluation^{*} \config{\mathemptyword}{\sigma'}$ holds,
 we have ${\sigma'}\equiv{\mathsubstbox{\sigma}{\mathsubst{x}{\sem{E}\sigma}}}$.
 Because of ${\mathsubstbox{\sigma}{\mathsubst{x}{\sem{E}\sigma}}} \models {Q}$, 
 \cref{lem:assert_subst} implies ${\sigma}\models{\mathsubstbox{Q}{\mathsubst{x}{E}}}$.
 
 \proofcase{\seqILrule}
 Assume that $\ihtriple{P}{C_{0}}{R}$ and $\ihtriple{R}{C_{1}}{Q}$ are valid.
 We show that $\ihtriple{P}{{C_{0}};{C_{1}}}{Q}$ is valid.

 Fix a state $\sigma'$ with ${\sigma'}\models{Q}$.
 Fix a state $\sigma$ with $\config{{C_{0}};{C_{1}}}{\sigma} \evaluation^{*} \config{\mathemptyword}{\sigma'}$.
 By \cref{lem:property-config} \cref{item:sequencing-lem-property-config},
 there exists $\sigma''$ such that
 ${\config{C_{0}}{\sigma}} \evaluation^{*} {\config{\mathemptyword}{\sigma''}}$ and
 ${\config{C_{1}}{\sigma''}} \evaluation^{*} {\config{\mathemptyword}{\sigma'}}$ hold.
 Since $\ihtriple{R}{C_{1}}{Q}$ is valid,
 we have ${\sigma''} \models {R}$.
 Since $\ihtriple{P}{C_{0}}{R}$ is valid,
 we have ${\sigma} \models {P}$.
 
 \proofcase{\consILrule}
 Assume that $\ihtriple{P}{C}{Q}$ is valid, and
 both ${P}\models{P'}$ and ${Q'}\models{Q}$ hold.
 We show that $\ihtriple{P'}{C}{Q'}$ is valid.
 
 Fix a state $\sigma'$ with ${\sigma'}\models{Q'}$.
 Fix a state $\sigma$ with $\config{C}{\sigma} \evaluation^{*} \config{\mathemptyword}{\sigma'}$.
 Because of ${Q'}\models{Q}$, we have ${\sigma'}\models{Q}$.
 Since $\ihtriple{P}{C}{Q}$ is valid, ${\sigma} \models {P}$.
 Because of ${P}\models{P'}$, we have ${\sigma}\models{P'}$.

 \proofcase{\orPILrule}
 Assume that $\ihtriple{P}{C_{0}}{Q}$ and $\ihtriple{P}{C_{1}}{Q}$ are valid.
 We show that $\ihtriple{P}{\cOR{C_{0}}{C_{1}}}{Q}$ is valid.

 Fix a state $\sigma'$ with ${\sigma'}\models{Q'}$.
 Fix a state $\sigma$ with 
 ${\config{\cOR{C_{0}}{C_{1}}}{\sigma}} \evaluation^{*} {\config{\mathemptyword}{\sigma'}}$.
 Assume that
 ${\config{\cOR{C_{0}}{C_{1}}}{\sigma}} \evaluation {\config{C_{0}}{\sigma}} \evaluation^{*} {\config{\mathemptyword}{\sigma'}}$.
 Since ${\config{C_{0}}{\sigma}} \evaluation^{*} {\config{\mathemptyword}{\sigma'}}$ holds and
 $\ihtriple{P}{C_{0}}{Q}$ is valid,
 we have ${\sigma} \models {P}$.

 In the same way,
 we have ${\sigma} \models {P}$ 
 if ${\config{\cOR{C_{0}}{C_{1}}}{\sigma}} \evaluation {\config{C_{1}}{\sigma}} \evaluation^{*} {\config{\mathemptyword}{\sigma'}}$ holds.
 Thus, $\ihtriple{P}{\cOR{C_{0}}{C_{1}}}{Q}$ is valid.

 \proofcase{\whilePILrule}
 Assume that $\ihtriple{{B}\to{P}}{C}{P}$ is valid.
 We show that $\ihtriple{P}{\while{B}{C}}{{\lnot B}\to{P}}$ is valid.

 Fix a state $\sigma'$ with ${\sigma'}\models{{\lnot B}\to{P}}$.
 Fix a state $\sigma$ with $\config{\while{B}{C}}{\sigma} \evaluation^{*} \config{\mathemptyword}{\sigma'}$.
 By \cref{lem:property-config} \cref{item:while-lem-property-config},
 there exist states $\sigma_{0}, \dots, \sigma_{k}$ 
 such that ${\sigma_{k}}\equiv{\sigma}$, ${\sigma_{0}}\equiv{\sigma'}$, 
 ${\sigma_{0}}\models{\lnot B}$, 
 ${\config{C}{\sigma_{i}}} \evaluation^{*} {\config{\mathemptyword}{\sigma_{i-1}}}$, and
 ${\sigma_{i}}\models{B}$
 hold for each $i=1, \dots, k$.
 We show ${\sigma_{i}}\models{P}$ for each $i=0, \dots, k$.
 The proof progresses by induction on $i$.

 Assume $i=0$.
 Then, ${\sigma_{0}}\equiv{\sigma'}$ holds.
 Since ${\sigma_{0}}\models{\lnot B}$ and ${\sigma_{0}}\models{{\lnot B}\to{P}}$ hold,
 we have ${\sigma_{0}}\models{P}$.

 Assume $i>0$.
 Then, we have ${\config{C}{\sigma_{i}}} \evaluation^{*} {\config{\mathemptyword}{\sigma_{i-1}}}$.
 By induction hypothesis, we have ${\sigma_{i-1}}\models{P}$.
 Since $\ihtriple{{B}\to{P}}{C}{P}$ is valid, we have ${\sigma_{i}}\models{{B}\to{P}}$.
 Because of ${\sigma_{i}}\models{B}$, we see ${\sigma_{i}}\models{P}$.

 Then, we have ${\sigma_{k}}\models{P}$.
 Since ${\sigma_{k}}\equiv{\sigma}$, we have ${\sigma}\models{P}$.
\end{proof}

Our proof system is relatively complete
if the expressiveness of the assertion language is sufficient, 
as in other Hoare-style logics (see \cite{Cook1978,Apt1981,Winskel1993,Vries2011,OHearn2019,Apt2O19,Lee2024}).
We say that \emph{the language of assertions is $\mathsetweakestpresy$-expressive}
if the following statement holds:
for any assertion $Q$ and any program $C$,
there exists an assertion $P$ such that 
${\sigma}\in{\mathsetweakestpre{C}{Q}}$ holds if and only if
${\sigma}\models{P}$ holds.
If a language of assertions includes some arithmetic operators, 
the language of assertions is $\mathsetweakestpresy$-expressive.
We give how to construct a weakest pre-condition assertion $\mathweakestpre{C}{Q}$
with some arithmetic operators in \cref{app:wpr}.

\begin{theorem}[Relative completeness]
\label[theorem]{thm:rev_hoare_complete}
 If the language of assertions is $\mathsetweakestpresy$-expressive,
 then any valid partial reverse Hoare triple $\ihtriple{P}{C}{Q}$ is provable 
 in \PIL.
\end{theorem}

In the remainder of this paper,
we assume that the language of assertions is $\mathsetweakestpresy$-expressive. 
For an assertion $Q$ and a program $C$,
we write ${\mathweakestpre{C}{Q}}$ for an assertion satisfying
the following condition:
${\sigma}\in{\mathsetweakestpre{C}{Q}}$ holds if and only if
${\sigma}\models{\mathweakestpre{C}{Q}}$ holds.
 
To show \cref{thm:rev_hoare_complete}, we show some lemmata.

\begin{lemma}
 \label[lemma]{lem:weakestness-assertion}
 $\ihtriple{P}{C}{Q}$ is valid
 if and only if 
 ${\mathweakestpre{C}{Q}}\models{P}$ holds.
\end{lemma}

\begin{proof}
 \noindent The `if' part:
 Assume ${\mathweakestpre{C}{Q}}\models{P}$.
 Fix $\sigma'$ with ${\sigma'}\models{Q}$.
 Fix a state $\sigma$ with $\config{C}{\sigma} \evaluation^{*} \config{\mathemptyword}{\sigma'}$.
 Then, ${\sigma}\in{\mathsetweakestpre{C}{Q}}$ holds.
 By the definition of $\mathweakestpre{C}{Q}$,
 ${\sigma}\models{\mathweakestpre{C}{Q}}$ holds.
 By assumption, ${\sigma}\models{P}$.
 Thus, $\ihtriple{P}{C}{Q}$ is valid.
 
 \noindent The `only if' part:
 Assume that $\ihtriple{P}{C}{Q}$ is valid.
 Let $\sigma$ with ${\sigma}\models{\mathweakestpre{C}{Q}}$.
 By the definition of $\mathweakestpre{C}{Q}$,
 ${\sigma}\in{\mathsetweakestpre{C}{Q}}$ holds.
 By \cref{prop:weakestness-set}, we have
 ${\mathsetweakestpre{C}{Q}}\subseteq{\mathsetintension{\sigma}{{\sigma}\models{P}}}$.
 Then, ${\sigma}\models{P}$ holds.
 Thus,  ${\mathweakestpre{C}{Q}}\models{P}$ holds.
\end{proof}

\begin{lemma}
 \label[lemma]{lem:thm_rev_hoare_complete}
 $\ihtriple{\mathweakestpre{C}{Q}}{C}{Q}$ is valid.
\end{lemma}

\begin{proof}
 By \cref{lem:weakestness-assertion}.
\end{proof}

\begin{lemma}
 \label[lemma]{lem:weakest-pre-predicate}
 Following statements hold:
 \begin{enumerate}
  \item ${\sigma}\models{\mathweakestpre{\mathemptyword}{Q}}$ holds
	if and only if ${\sigma}\models{Q}$ holds.
	\label{item:emptyword-lem-weakest-post-predicate}
  \item ${\sigma}\models{\mathweakestpre{\assign{x}{E}}{Q}}$ holds
	if and only if 
	${\sigma}\models{\mathsubstbox{Q}{\mathsubst{x}{E}}}$ holds.
	\label{item:assignment-lem-weakgest-post-predicate}
  \item ${\sigma}\models{\mathweakestpre{C_{0};C_{1}}{Q}}$ holds
	if and only if 
	${\sigma}\models{\mathweakestpre{C_{0}}{\mathweakestpre{C_{1}}{Q}}}$ holds.
	\label{item:cons-lem-strongest-post-predicate}
  \item ${\sigma}\models{\mathweakestpre{\cOR{C_{0}}{C_{1}}}{Q}}$ holds
	if and only if 
	${\sigma}\models{{\mathweakestpre{C_{0}}{Q}}\lor{\mathweakestpre{C_{1}}{Q}}}$ holds.
	\label{item:cor-lem-strongest-post-predicate}
 \end{enumerate}
\end{lemma}

\begin{proof}
 We show each statement.

 \noindent \cref{item:emptyword-lem-weakest-post-predicate} 
 \noindent The `if' part: Assume ${\sigma}\models{Q}$.
 Because of ${\config{\mathemptyword}{\sigma} \evaluation^{*} \config{\mathemptyword}{\sigma}}$,
 we have ${\sigma}\models{\mathweakestpre{\mathemptyword}{Q}}$.

 \noindent The `only if' part:
 Assume ${\sigma}\models{\mathweakestpre{\mathemptyword}{Q}}$.
 By definition of ${\mathweakestpre{\mathemptyword}{Q}}$, we have
 ${\config{\mathemptyword}{\sigma} \evaluation^{*} \config{\mathemptyword}{\sigma}}$ and
 ${{\sigma} \models {Q}}$.
 Then, we see ${{\sigma} \models {Q}}$.

 \noindent \cref{item:assignment-lem-weakgest-post-predicate}
 \noindent The `if' part: Assume ${\sigma}\models{\mathsubstbox{Q}{\mathsubst{x}{E}}}$.
 Then, we have 
 ${\config{\assign{x}{E}}{\sigma} \evaluation \config{\mathemptyword}{\mathsubstbox{\sigma}{\mathsubst{x}{\sem{E}\sigma}}}}$.
 Because of ${\sigma}\models{\mathsubstbox{Q}{\mathsubst{x}{E}}}$,
 we have ${{\mathsubstbox{\sigma}{\mathsubst{x}{\sem{E}\sigma}}} \models {Q}}$.
 Thus, ${\sigma}\models{\mathweakestpre{C}{Q}}$ holds.

 \noindent The `only if' part: Assume ${\sigma}\models{\mathweakestpre{\assign{x}{E}}{Q}}$.
 Then, we have
 ${\config{\assign{x}{E}}{\sigma} \evaluation \config{\mathemptyword}{\mathsubstbox{\sigma}{\mathsubst{x}{\sem{E}\sigma}}}}$ and
 ${{\mathsubstbox{\sigma}{\mathsubst{x}{\sem{E}\sigma}}} \models {Q}}$.
 By \cref{lem:assert_subst},
 ${\sigma} \models {\mathsubstbox{Q}{\mathsubst{x}{E}}}$ holds.

 \noindent \cref{item:cons-lem-strongest-post-predicate}
 \noindent The `if' part:
 Assume ${\sigma}\models{\mathweakestpre{C_{0}}{\mathweakestpre{C_{1}}{Q}}}$.
 Hence, there exists $\sigma''$ such that 
 ${\config{C_{0}}{\sigma}} \evaluation^{*} {\config{\mathemptyword}{\sigma''}}$ and 
 ${\sigma''}\models{\mathweakestpre{C_{1}}{Q}}$ hold.
 Therefore, there exists a state $\sigma'$ such that
 ${\config{C_{1}}{\sigma''}} \evaluation^{*} {\config{\mathemptyword}{\sigma'}}$ and
 ${{\sigma'} \models {Q}}$.
 By \cref{lem:property-config} \cref{item:sequencing-lem-property-config},
 ${\config{C_{0}; C_{1}}{\sigma} \evaluation^{*} \config{\mathemptyword}{\sigma'}}$ holds.
 Thus, ${\sigma}\in{\mathsetweakestpre{C}{Q}}$ holds.

 \noindent The `only if' part:
 Assume ${\sigma}\models{\mathweakestpre{C_{0};C_{1}}{Q}}$.
 Then, there exists a state $\sigma'$ such that
 ${\config{C_{0}; C_{1}}{\sigma} \evaluation^{*} \config{\mathemptyword}{\sigma'}}$ and
 ${{\sigma'} \models {Q}}$ hold.
 By \cref{lem:property-config} \cref{item:sequencing-lem-property-config},
 there exists $\sigma''$ such that 
 ${\config{C_{0}}{\sigma}} \evaluation^{*} {\config{\mathemptyword}{\sigma''}}$ and 
 ${\config{C_{1}}{\sigma''}} \evaluation^{*} {\config{\mathemptyword}{\sigma'}}$ hold.
 Since ${\config{C_{1}}{\sigma''}} \evaluation^{*} {\config{\mathemptyword}{\sigma'}}$ and
 ${{\sigma'} \models {Q}}$ hold,
 we have ${\sigma''}\in{\mathsetweakestpre{C_{1}}{Q}}$.
 Then, we have ${\sigma''}\models{\mathweakestpre{C_{1}}{Q}}$.
 Because of ${\config{C_{0}}{\sigma}} \evaluation^{*} {\config{\mathemptyword}{\sigma''}}$,
 we see ${\sigma}\in{\mathsetweakestpre{C_{0}}{\mathweakestpre{C_{1}}{Q}}}$.
 Then, we have ${\sigma}\models{\mathweakestpre{C_{0}}{\mathweakestpre{C_{1}}{Q}}}$.

 \noindent \cref{item:cor-lem-strongest-post-predicate}
 \noindent The `if' part:
 Assume ${\sigma}\models{{\mathweakestpre{C_{0}}{Q}}\lor{\mathweakestpre{C_{1}}{Q}}}$.
 Then, either ${\sigma}\models{\mathweakestpre{C_{0}}{Q}}$ or 
 ${\sigma}\models{\mathweakestpre{C_{1}}{Q}}$ holds.

 Assume ${\sigma}\models{\mathweakestpre{C_{0}}{Q}}$.
 Then, there exists a state $\sigma'$ such that
 $\config{C_{0}}{\sigma} \evaluation^{*} \config{\mathemptyword}{\sigma'}$ and
 ${{\sigma'} \models {Q}}$ hold.
 Because $\config{\cOR{C_{0}}{C_{1}}}{\sigma} \evaluation \config{C_{0}}{\sigma}$ holds,
 we have $\config{\cOR{C_{0}}{C_{1}}}{\sigma} \evaluation^{*} \config{\mathemptyword}{\sigma'}$.
 Hence, we have ${\sigma}\models{\mathweakestpre{\cOR{C_{0}}{C_{1}}}{Q}}$.

 In the similar way, we have ${\sigma}\models{\mathweakestpre{\cOR{C_{0}}{C_{1}}}{Q}}$
 if ${\sigma}\models{\mathweakestpre{C_{1}}{Q}}$ holds.

 \noindent The `only if' part:
 Assume ${\sigma}\models{\mathweakestpre{\cOR{C_{0}}{C_{1}}}{Q}}$.
 Then, there exists a state $\sigma'$ such that
 $\config{\cOR{C_{0}}{C_{1}}}{\sigma} \evaluation^{*} \config{\mathemptyword}{\sigma'}$ and
 ${{\sigma'} \models {Q}}$ hold.
 We see either 
 $\config{\cOR{C_{0}}{C_{1}}}{\sigma} \evaluation \config{C_{0}}{\sigma} \evaluation^{*} \config{\mathemptyword}{\sigma'}$
 or 
 $\config{\cOR{C_{0}}{C_{1}}}{\sigma} \evaluation \config{C_{1}}{\sigma} \evaluation^{*} \config{\mathemptyword}{\sigma'}$ 
 holds.
 Hence, we have either ${\sigma}\models{\mathweakestpre{C_{0}}{Q}}$ or 
 ${\sigma}\models{\mathweakestpre{C_{1}}{Q}}$.
 Thus, ${\sigma}\models{{\mathweakestpre{C_{0}}{Q}}\lor{\mathweakestpre{C_{1}}{Q}}}$ holds.

\end{proof}

Now, we show \cref{thm:rev_hoare_complete}.

\begin{proof}[Proof of \cref{thm:rev_hoare_complete}]
 Assume that $\ihtriple{P}{C}{Q}$ is valid.
 We show $\ihtriple{P}{C}{Q}$ is provable in \PIL.
 The proof is by induction on construction of $C$.

 \proofcase{(${C}\equiv{\mathemptyword}$)} 
 Since $\ihtriple{P}{C}{Q}$ is valid,
 \cref{lem:weakestness-assertion} and 
 \cref{lem:weakest-pre-predicate} \cref{item:emptyword-lem-weakest-post-predicate} 
 imply that ${Q}\models{P}$.
 We have a proof of $\ihtriple{P}{C}{Q}$ as follows:
 \begin{center}
  \begin{inlineprooftree}
   \AxiomC{}
   \RightLabel{\axiomILrule}
   \UnaryInfC{$\ihtriple{Q}{\varepsilon}{Q}$}
   \LeftLabel{(${Q}\models{P}$)}
   \RightLabel{\consILrule}
   \UnaryInfC{$\ihtriple{P}{C}{Q}$}
  \end{inlineprooftree}.
 \end{center}

 \proofcase{(${C}\equiv{\assign{x}{E}}$)} 
 Since $\ihtriple{P}{C}{Q}$ is valid,
 \cref{lem:weakestness-assertion} and 
 \cref{lem:weakest-pre-predicate} \cref{item:assignment-lem-weakgest-post-predicate}
 imply that ${\mathsubstbox{Q}{\mathsubst{x}{E}}}\models{P}$.
 We have a proof of $\ihtriple{P}{C}{Q}$ as follows:
 \begin{center}
  \begin{inlineprooftree}
  \AxiomC{\phantom{$\ihtriple{\mathsubstbox{Q}{\mathsubst{x}{E}}}{\assign{x}{E}}{Q}$}}
  \RightLabel{\assignPILrule}
  \UnaryInfC{$\ihtriple{\mathsubstbox{Q}{\mathsubst{x}{E}}}{\assign{x}{E}}{Q}$}
  \RightLabel{\consILrule}
  \UnaryInfC{$\ihtriple{P}{C}{Q}$}
  \end{inlineprooftree}.
 \end{center}

 \proofcase{(${C}\equiv{C_{0};C_{1}}$)} 
 Since $\ihtriple{P}{C}{Q}$ is valid,
 \cref{lem:weakestness-assertion} and
 \cref{lem:weakest-pre-predicate} \cref{item:cons-lem-strongest-post-predicate}
 imply that ${\mathweakestpre{C_{0}}{\mathweakestpre{C_{1}}{Q}}}\models{P}$.
 By \cref{lem:thm_rev_hoare_complete}, \linebreak[3]
 $\ihtriple{\mathweakestpre{C_{1}}{Q}}{C_{1}}{Q}$ is valid. \linebreak[3]
 By induction hypothesis, \linebreak[3]
 $\ihtriple{\mathweakestpre{C_{1}}{Q}}{C_{1}}{Q}$ is provable. \linebreak[3]
 By \cref{lem:thm_rev_hoare_complete}, \linebreak[3]
 $\ihtriple{\mathweakestpre{C_{0}}{\mathweakestpre{C_{1}}{Q}}}{C_{0}}{\mathweakestpre{C_{1}}{Q}}$ is valid. \linebreak[3]
 By \linebreak[3] induction \linebreak[3] hypothesis, \linebreak[3]
 $\ihtriple{\mathweakestpre{C_{0}}{\mathweakestpre{C_{1}}{Q}}}{C_{0}}{\mathweakestpre{C_{1}}{Q}}$ is provable.

 Let $\pi_{0}$ be a proof of 
 $\ihtriple{\mathweakestpre{C_{0}}{\mathweakestpre{C_{1}}{Q}}}{C_{0}}{\mathweakestpre{C_{1}}{Q}}$.
 Let $\pi_{1}$ be a proof of $\ihtriple{\mathweakestpre{C_{1}}{Q}}{C_{1}}{Q}$.
 Then, we have a proof of $\ihtriple{P}{C}{Q}$ as follows:
 \begin{center}
  \begin{inlineprooftree}
   \AxiomC{}
   \RightLabel{${\pi_{0}}$}
   \DeduceC{$\ihtriple{\mathweakestpre{C_{0}}{\mathweakestpre{C_{1}}{Q}}}{C_{0}}{\mathweakestpre{C_{1}}{Q}}$}

   \AxiomC{}
   \RightLabel{${\pi_{1}}$}
   \DeduceC{$\ihtriple{\mathweakestpre{C_{1}}{Q}}{C_{1}}{Q}$}
   \RightLabel{\seqILrule}
   \BinaryInfC{$\ihtriple{\mathweakestpre{C_{0}}{\mathweakestpre{C_{1}}{Q}}}{{C_{0}};{C_{1}}}{Q}$}
   \LeftLabel{(${\mathweakestpre{C_{0}}{\mathweakestpre{C_{1}}{Q}}}\models{P}$)}
   \RightLabel{\consILrule}
   \UnaryInfC{$\ihtriple{P}{C}{Q}$}
  \end{inlineprooftree}.
 \end{center}

 \proofcase{(${C}\equiv{\cOR{C_{0}}{C_{1}}}$)} 
 Since $\ihtriple{P}{C}{Q}$ is valid,
 \cref{lem:weakestness-assertion} and
 \cref{lem:weakest-pre-predicate} \cref{item:cor-lem-strongest-post-predicate}
 imply that ${{\mathweakestpre{C_{0}}{Q}}\lor{\mathweakestpre{C_{1}}{Q}}}\models{P}$.
 By \cref{lem:thm_rev_hoare_complete}, \linebreak[3]
 $\ihtriple{\mathweakestpre{C_{i}}{Q}}{C_{i}}{Q}$ is valid for ${i}={0, 1}$. \linebreak[3]
 By induction hypothesis, \linebreak[3]
 $\ihtriple{\mathweakestpre{C_{i}}{Q}}{C_{i}}{Q}$ is provable for ${i}={0, 1}$. 
 \linebreak[3]
 Let $\pi_{i}$ be a proof of $\ihtriple{\mathweakestpre{C_{i}}{Q}}{C_{i}}{Q}$ 
 for ${i}={0, 1}$. 
 Then, we have a proof of $\ihtriple{P}{C}{Q}$ as follows:
 \begin{center}
  \begin{inlineprooftree}
   \AxiomC{}
   \RightLabel{${\pi_{0}}$}
   \DeduceC{$\ihtriple{\mathweakestpre{C_{0}}{Q}}{C_{0}}{Q}$}
   \RightLabel{\consILrule}
   \UnaryInfC{$\ihtriple{{\mathweakestpre{C_{0}}{Q}}\lor{\mathweakestpre{C_{1}}{Q}}}{C_{0}}{Q}$}
   
   \AxiomC{}
   \RightLabel{${\pi_{1}}$}
   \DeduceC{$\ihtriple{\mathweakestpre{C_{1}}{Q}}{C_{1}}{Q}$}
   \RightLabel{\consILrule}
   \UnaryInfC{$\ihtriple{{\mathweakestpre{C_{0}}{Q}}\lor{\mathweakestpre{C_{1}}{Q}}}{C_{1}}{Q}$}
   \RightLabel{\orPILrule}
   \BinaryInfC{$\ihtriple{{\mathweakestpre{C_{0}}{Q}}\lor{\mathweakestpre{C_{1}}{Q}}}{\cOR{C_{0}}{C_{1}}}{Q}$}
   \LeftLabel{(${{\mathweakestpre{C_{0}}{Q}}\lor{\mathweakestpre{C_{1}}{Q}}}\models{P}$)}
   \RightLabel{\consILrule}
   \UnaryInfC{$\ihtriple{P}{C}{Q}$}
  \end{inlineprooftree}.
 \end{center}

 \proofcase{(${C}\equiv{\while{B}{C_{0}}}$)}
 To show that $\ihtriple{P}{C}{Q}$ is provable in \PIL,
 we show that \linebreak[3] 
 $\ihtriple{{B}\to{\mathweakestpre{C}{Q}}}{C_{0}}{\mathweakestpre{C}{Q}}$ 
 is valid, and ${Q}\models {{\lnot B} \to {\mathweakestpre{C}{Q}}}$ holds. 

 
 We show that 
 $\ihtriple{{B}\to{\mathweakestpre{C}{Q}}}{C_{0}}{\mathweakestpre{C}{Q}}$ is valid.
 Fix a state $\sigma'$ with ${\sigma'}\models{\mathweakestpre{C}{Q}}$.
 Fix a state $\sigma$ with 
 ${\config{C_{0}}{\sigma}} \evaluation^{*} {\config{\mathemptyword}{\sigma'}}$.
 We show ${\sigma}\models{{B}\to{\mathweakestpre{C}{Q}}}$. 
 Assume ${\sigma}\models{B}$.
 Then, we have
 \[
 {\config{\while{B}{C_{0}}}{\sigma}}\evaluation{\config{{C_{0}};{\while{B}{C_{0}}}}{\sigma}}.
 \]
 By ${\config{C_{0}}{\sigma}} \evaluation^{*} {\config{\mathemptyword}{\sigma'}}$, we have 
\[
 {\config{{C_{0}};{\while{B}{C_{0}}}}{\sigma}} \evaluation^{*} {\config{\while{B}{C_{0}}}{\sigma'}}.
\]
 Because of ${\sigma'}\models{\mathweakestpre{C}{Q}}$,
 there exists $\sigma''$ such that
 ${\config{C}{\sigma'}} \evaluation^{*} {\config{\mathemptyword}{\sigma''}}$
 and ${\sigma''}\models{Q}$ hold.
 Then, we have
 ${\config{C}{\sigma}} \evaluation^{*} {\config{C}{\sigma'}} \evaluation^{*} {\config{\mathemptyword}{\sigma''}}$
 and ${\sigma''}\models{Q}$.
 Hence, 
 we have
 ${\sigma}\models{\mathweakestpre{C}{Q}}$.
 Therefore, ${\sigma}\models{{B}\to{\mathweakestpre{C}{Q}}}$ holds. 
 Thus, $\ihtriple{{B}\to{\mathweakestpre{C}{Q}}}{C_{0}}{\mathweakestpre{C}{Q}}$ is valid.

 We show ${Q}\models {{\lnot B} \to {\mathweakestpre{C}{Q}}}$.
 Fix a state $\sigma$ with ${\sigma}\models{Q}$.
 Assume ${\sigma}\models{\lnot B}$.
 Then, we have
 \[
 {\config{\while{B}{C_{0}}}{\sigma}}\evaluation{\config{\mathemptyword}{\sigma}}.
 \]
 Hence, 
 we have ${\sigma}\models{\mathweakestpre{C}{Q}}$.
 Thus, ${Q}\models {{\lnot B} \to {\mathweakestpre{C}{Q}}}$.

 Because $\ihtriple{P}{C}{Q}$ is valid, 
 \cref{lem:weakestness-assertion} implies ${\mathweakestpre{C}{Q}}\models {P}$.
 Since $\ihtriple{{B}\to{\mathweakestpre{C}{Q}}}{C_{0}}{\mathweakestpre{C}{Q}}$ is valid,
 induction hypothesis implies that
 $\ihtriple{{B}\to{\mathweakestpre{C}{Q}}}{C_{0}}{\mathweakestpre{C}{Q}}$ is provable.
 Let $\pi$ be a proof of $\ihtriple{{B}\to{\mathweakestpre{C}{Q}}}{C_{0}}{\mathweakestpre{C}{Q}}$.
 Then, we have a proof of $\ihtriple{P}{C}{Q}$ as follows:
 \begin{center}
  \begin{inlineprooftree}
   \AxiomC{}
   \RightLabel{${\pi}$}
   \DeduceC{$\ihtriple{{B}\to{\mathweakestpre{C}{Q}}}{C_{0}}{\mathweakestpre{C}{Q}}$}

  \RightLabel{\whilePILrule}
  \UnaryInfC{$\ihtriple{\mathweakestpre{C}{Q}}{\while{B}{C_{0}}}{{\lnot B}\to{\mathweakestpre{C}{Q}}}$}
   \RightLabel{\consILrule}
   \UnaryInfC{$\ihtriple{P}{C}{Q}$}
  \end{inlineprooftree}.
 \end{center}
\end{proof}
\section{Cyclic proofs for partial incorrectness logic (partial reverse Hoare logic)}
\label{sec:cyclic-proofs-PRHL}
This section introduces cyclic proofs for partial reverse Hoare logic.

In our ordinary proof system, given in \cref{sec:partial-reverse-Hoare-logic},
we have to find a good loop invariant when \whilePILrule\ is applied.
However, it is challenged to find a suitable loop invariant \cite{Furia2014}.
In contrast, our cyclic proofs do not have to find any loop invariants 
when the rule for the while loop is applied.
This point is an advantage of cyclic proofs from the view of proof search.

\cref{fig:rules-for-cyclic-proofs-of-par-incor} shows 
the inference rules for cyclic proofs.
To contain cycles, the form of rules are changed from that of the ordinary proof system.
We note that $P$ and $Q$ in the rule for the while loop \whileCPILrule\ can be arbitrary.
In other words, we do not have to find any loop invariants when \whileCPILrule\ is applied.

\begin{figure*}[t]
 \centering
 \begin{inlineprooftree}
  \AxiomC{}
  \RightLabel{\axiomILrule}
  \UnaryInfC{$\ihtriple{Q}{\mathemptyword}{Q}$}
 \end{inlineprooftree}
 \begin{inlineprooftree}
  \AxiomC{$\ihtriple{P'}{C}{Q'}$}
  \LeftLabel{(${P'}\models{P}$, ${Q}\models{Q'}$)}
  \RightLabel{\consILrule}
  \UnaryInfC{$\ihtriple{P}{C}{Q}$}
 \end{inlineprooftree}
 \begin{inlineprooftree}
  \AxiomC{$\ihtriple{{{x'}={\mathsubstbox{E}{\mathsubst{x}{x'}}}}\land{\mathsubstbox{P}{\mathsubst{x}{x'}}}}{C}{Q}$}
  \RightLabel{\assignCPILrule}
  \UnaryInfC{$\ihtriple{P}{{\assign{x}{E}};{C}}{Q}$}
 \end{inlineprooftree}
 \begin{inlineprooftree}
  \AxiomC{$\ihtriple{P}{C_{0};C}{Q}$}
  \AxiomC{$\ihtriple{P}{C_{1};C}{Q}$}
  \RightLabel{\orCPILrule}
  \BinaryInfC{$\ihtriple{P}{\cOR{C_{0}}{C_{1}};C}{Q}$}
 \end{inlineprooftree}
 \begin{inlineprooftree}
  \AxiomC{$\ihtriple{{\lnot B}\to{P}}{C'}{Q}$}
  \AxiomC{$\ihtriple{{B}\to{P}}{C;{\while{B}{C}};{C'}}{Q}$}
  \RightLabel{\whileCPILrule}
  \BinaryInfC{$\ihtriple{P}{{\while{B}{C}};{C'}}{Q}$}
 \end{inlineprooftree}
 
 \caption{Rules for cyclic proofs of partial incorrectness logic (partial reverse Hoare logic)}
 \label{fig:rules-for-cyclic-proofs-of-par-incor}
\end{figure*}

\begin{definition}[Cyclic proofs for reverse Hoare logic  (\CPIL -proof)]
 \label[definition]{definition:CPIL_proof}
 A \emph{leaf} of a derivation tree constructed 
 according to the proof rules in \cref{fig:rules-for-cyclic-proofs-of-par-incor}
 is said to be \emph{open} 
 if it is not the conclusion of \axiomILrule.
 A \emph{companion} of a leaf in a derivation tree is defined as
 an inner node of the derivation tree labelled by the same triple as the leaf label.

 A \emph{\CPIL -pre-proof} is defined as a pair $\mathcal{P} = (\mathcal{D,L})$, 
 where $\mathcal{D}$ is a finite derivation tree constructed 
 according to the proof rules in \cref{fig:rules-for-cyclic-proofs-of-par-incor} and 
 $\mathcal{L}$ is a \emph{back-link function} that maps each open leaf of $\mathcal{D}$
 to one of its companions.

 A \CPIL -pre-proof $\mathprooffig{P}$ is called a \emph{\CPIL -proof} 
 if it satisfies the following \emph{global soundness condition}: 
 the rules except for \consILrule\ are applied infinitely many often 
 along each infinite path in $\mathprooffig{P}$.
 If there is a \CPIL -proof whose root is labelled by $\ihtriple{P}{C}{Q}$,
 we say that $\ihtriple{P}{C}{Q}$ is \emph{provable in \CPIL}.
\end{definition}

We note that some \CPIL -pre-proofs are not finite trees
because cycles are allowed in cyclic proofs.
However, each \CPIL -pre-proof can be understood 
as a \emph{regular} (possibly infinite) tree
whose subtrees are finitely many.

\begin{example}
The following is a \CPIL -proof:
\begin{center} 
 \begin{inlineprooftree}
  \AxiomC{}
  \RightLabel{\axiomILrule}
  \UnaryInfC{$\ihtriple{Q}{\mathemptyword}{Q}$}
  \RightLabel{\assignCPILrule}
  \UnaryInfC{$\ihtriple{{x=10}\land {i=4}}{{\assign{i}{i+1}}}{Q}$}
  \RightLabel{\assignCPILrule}
  \UnaryInfC{$\ihtriple{{x=6}\land {i=4}}{C_{0}}{Q}$}
  \RightLabel{\consILrule}
  \UnaryInfC{$\ihtriple{{\lnot (i<5)}\to {P}}{C_{0}}{Q}$}

  \AxiomC{$\ihtriple{P}{C}{Q}$ \raisebox{1ex}{\tikzmark{bud}}}
  \RightLabel{\consILrule}
  \UnaryInfC{$\ihtriple{{i<6}\to{{x+1=i}\land{i=1}}}{C}{Q}$}
  \RightLabel{\assignCPILrule}
  \UnaryInfC{$\ihtriple{{i<5}\to{{x=i}\land{i=0}}}{{\assign{i}{i+1}};C}{Q}$}
  \RightLabel{\assignCPILrule \tikzmark{mid}}
  \UnaryInfC{$\ihtriple{{i<5}\to{P}}{C_{0};C}{Q}$}

  \RightLabel{\whileCPILrule}
  \BinaryInfC{$\ihtriple{P}{C}{Q}$ \tikzmark{comp}}
 \end{inlineprooftree},
\end{center}
 \begin{tikzpicture}[remember picture, overlay, thick, relative, auto, rounded corners, line width=1.0pt]
    \coordinate (b) at ({pic cs:bud});
    \coordinate (root) at ({pic cs:comp});
    \coordinate (M) at ({pic cs:mid});
    
    \draw[->] let \p1=($(M)-(b)$), \p2=($(root)-(b)$) in
    (b) -- ++(\x1 + .25em,0) -- ++(0,\y2) -- (root);
 \end{tikzpicture}
 where
 ${C}\equiv{\while{i<5}{C_{0}}}$,
 ${C_{0}}\equiv{{\assign{x}{x+i}};{\assign{i}{i+1}}}$,
 ${P}\equiv{{x=0}\land{i=0}}$ and
 ${Q}\equiv{{x=10}\land {i=5}}$, and 
 the arrow indicates the pairing of the companion with the bud.
 We see that the global soundness condition holds, immediately.
 When we apply \whileCPILrule\ in the root,
 we do not find any loop invariant.
\end{example}

Now, we show the soundness of cyclic proofs.
To show the soundness, we show a lemma.
\begin{lemma}
 \label[lemma]{lem:local_soundness_CPIL}
 Each of the proof rules in \cref{fig:rules-for-cyclic-proofs-of-par-incor} 
 has the following property: 
 Suppose the conclusion of the rule $\ihtriple{P}{C}{Q}$ is not valid, 
 so that in particular
 there exist a natural number $n$ and states $\sigma$, $\sigma'$ such that
 $\sigma' \models Q$, $\config{C}{\sigma} \evaluation^{n} \config{\mathemptyword}{\sigma'}$, and
 ${\sigma} \not\models {P}$ hold. 
 Then, for some premise of the rule $\ihtriple{P'}{C'}{Q'}$,
 $\sigma' \models Q'$ holds and
 there exist a natural number $n'$ and a state $\sigma''$ such that $n' \leq n$,
 $\sigma'' \not\models P'$, and
 $\config{C'}{\sigma''} \evaluation^{n'} \config{\mathemptyword}{\sigma'}$ hold. 
 Moreover, for all rules except \consILrule, we have $n' < n$.
\end{lemma}

\begin{proof}
 We show the statement for each rule. 
 
 \proofcase{\axiomILrule} 
 Since the conclusion of the rule \axiomILrule\ is always valid,
 the assumption does not hold in this case.

 \proofcase{\assignCPILrule}
 Let ${P}\equiv{\mathsubstbox{P'}{\mathsubst{x}{E}}}$ and
 ${C}\equiv{{\assign{x}{E}};{C'}}$.
 Assume that there exists a natural number $n$ and states $\sigma$, $\sigma'$ such that
 $\sigma' \models Q$, $\config{C}{\sigma} \evaluation^{n} \config{\mathemptyword}{\sigma'}$, and
 ${\sigma} \not\models {\mathsubstbox{P'}{\mathsubst{x}{E}}}$ hold.

 We show that
 there exist a natural number $n'$ and a state $\sigma''$ such that $n' < n$,
 $\sigma'' \not\models P'$, and
 $\config{C'}{\sigma''} \evaluation^{n'} \config{\mathemptyword}{\sigma'}$ hold. 

 Let ${\sigma''}\equiv{\mathsubstbox{\sigma}{\mathsubst{x}{\sem{E}}\sigma}}$. 
 Then, we have
 \[
 \config{C}{\sigma} \evaluation \config{C'}{\sigma''} \evaluation^{n-1} \config{\mathemptyword}{\sigma'}.
 \]
 Hence, $n'$ is $n-1$.
 Because of ${\sigma} \not\models {\mathsubstbox{P'}{\mathsubst{x}{E}}}$,
 we have ${\mathsubstbox{\sigma}{\mathsubst{x}{\sem{E}}\sigma}} \not\models {P'}$.
 Hence, $\sigma'' \not\models P'$ holds.

 \proofcase{\consILrule} 
 Assume ${P'}\models{P}$ and ${Q}\models{Q'}$.
 We also assume that 
 there exist a natural number $n$ and states $\sigma$, $\sigma'$ such that
 $\sigma' \models Q$, $\config{C}{\sigma} \evaluation^{n} \config{\mathemptyword}{\sigma'}$, and
 ${\sigma} \not\models {P}$ hold. 

 We show that
 $\sigma' \models Q'$ holds and
 there exists  a state $\sigma''$ such that
 $\sigma'' \not\models P'$ and
 $\config{C}{\sigma''} \evaluation^{n} \config{\mathemptyword}{\sigma'}$ hold. 

 Because of ${Q}\models{Q'}$ and $\sigma' \models Q$,
 we have $\sigma' \models Q'$.

 Let ${\sigma''}\equiv{\sigma}$.
 Then, we have
 \[
 \config{C}{\sigma''} \evaluation^{n} \config{\mathemptyword}{\sigma'}.
 \]
 Because of ${P'}\models{P}$ and ${\sigma''} \not\models {P}$,
 we have $\sigma'' \not\models P'$.

 \proofcase{\whileCPILrule} 
 Let ${C}\equiv{{\while{B}{C_{0}}};{C'}}$.
 Assume that there exist a natural number $n$ and states $\sigma$, $\sigma'$ such that
 $\sigma' \models Q$, $\config{C}{\sigma} \evaluation^{n} \config{\mathemptyword}{\sigma'}$, and
 ${\sigma} \not\models {P}$ hold. 

 We show that there exist a natural number $n'$ and a state $\sigma''$ such that $n' < n$,
 and
 either both $\sigma'' \not\models {{\lnot B}\to{P}}$ and
 $\config{C'}{\sigma''} \evaluation^{n'} \config{\mathemptyword}{\sigma'}$,
 or
 both $\sigma'' \not\models {{B}\to{P}}$ and
 $\config{{C_{0};{\while{B}{C_{0}}};{C'}}}{\sigma''} \evaluation^{n'} \config{\mathemptyword}{\sigma'}$
 hold. 
 Let ${\sigma''}\equiv{\sigma}$.

 Assume ${\sigma}\models{B}$.
 Then, we have
 \[
 \config{C}{\sigma} \evaluation \config{C_{0};{\while{B}{C_{0}}};{C'}}{\sigma''} \evaluation^{n-1} \config{\mathemptyword}{\sigma'}.
 \]
 Hence, $n'$ is $n-1$.
 Since ${\sigma} \not\models {P}$ and ${\sigma''}\models{B}$ holds,
 we have ${\sigma''} \not\models {{B}\to{P}}$.

 Assume ${\sigma}\models{\lnot B}$.
 Then, we have
 \[
 \config{C}{\sigma} \evaluation \config{C'}{\sigma''} \evaluation^{n-1} \config{\mathemptyword}{\sigma'}.
 \]
 Hence, $n'$ is $n-1$.
 Since ${\sigma} \not\models {P}$ and ${\sigma''}\models{\lnot B}$ holds,
 we have ${\sigma''} \not\models {{\lnot B}\to{P}}$.
\end{proof}

We show the soundness theorem.

\begin{theorem}[Soundness]
\label[theorem]{thm:par-incorrect_soundness}
 If there is a \CPIL -proof of $\ihtriple{P}{C}{Q}$, then it is valid.
\end{theorem}

\begin{proof}
 Assume, for contradiction, there is a \CPIL -proof of $\ihtriple{P}{C}{Q}$, 
 but it is not valid.
 Then, there exist a natural number $n$ and states $\sigma$, $\sigma'$ such that
 $\sigma' \models Q$, $\config{C}{\sigma} \evaluation^{n} \config{\mathemptyword}{\sigma'}$, and
 ${\sigma} \not\models {P}$ hold. 

 We inductively define an infinite path 
 $\mathsequence{\ihtriple{P_{i}}{C_{i}}{Q_{i}}}{{i}\geq {0}}$ 
 in the \CPIL -proof of $\ihtriple{P}{C}{Q}$ and
 an infinitely non-increasing sequence of natural numbers $\mathsequence{n_{i}}{i\geq 0}$
 satisfying the following conditions:
 for ${\ihtriple{P_{i}}{C_{i}}{Q_{i}}}$,
 there exist states $\sigma_{i}$ and $\sigma'_{i}$ such that
 ${\sigma_{i}}\models{P_{i}}$, 
 $\config{C_{i}}{\sigma_{i}} \evaluation^{n_{i}} \config{\mathemptyword}{\sigma'_{i}}$, 
 and ${\sigma'_{i}}\not\models {Q_{i}}$.
 
 Define ${\ihtriple{P_{0}}{C_{0}}{Q_{0}}}$ as ${\ihtriple{P}{C}{Q}}$.

 Assume that ${\ihtriple{P_{i-1}}{C_{i-1}}{Q_{i-1}}}$ is defined.
 By the condition, ${\ihtriple{P_{i-1}}{C_{i-1}}{Q_{i-1}}}$ is not valid.
 By \cref{lem:local_soundness_CPIL},
 there exist states $\sigma_{i}$ and $\sigma'_{i}$ such that
 ${\sigma_{i}}\models{P_{i}}$, 
 $\config{C_{i}}{\sigma_{i}} \evaluation^{n_{i}} \config{\mathemptyword}{\sigma'_{i}}$, 
 and ${\sigma'_{i}}\not\models {Q_{i}}$.
 If it is a symbolic execution rule, then $n_{i} < n_{i-1}$ 
 i.e. the length of the computations is monotonously decreasing.

 By the global soundness condition on \CPIL -proofs, 
 every infinite path has rules except for \consILrule\ applied infinitely often. 
 By \cref{lem:local_soundness_CPIL}, 
 $\mathsequence{n_{i}}{i\geq 0}$ is an infinite descending sequence
 of natural numbers.
 This is a contradiction, and we conclude that $\ihtriple{P}{C}{Q}$ is valid after all.
\end{proof}

In the remainder of this section, we show the relative completeness of \CPIL, i.e.
the provability of cyclic proofs is the same as that of our ordinary proof system \PIL.
We show this statement by giving the way to transform each \PIL -proof into a \CPIL -proof.
The formal statement of the completeness is the following.
\begin{theorem}[Relative completeness of \CPIL]
 \label[theorem]{thm:cyclic_par_incor_complete}
 For any partial reverse Hoare triple $\ihtriple{P}{C}{Q}$, 
 the following statements are equivalent:
 \begin{enumerate}
  \item $\ihtriple{P}{C}{Q}$ is valid.  \label{item:validness-thm-cyclic_par_incor_complete}
  \item $\ihtriple{P}{C}{Q}$ is provable in \PIL. \label{item:provable-reverse_Hoare-thm-cyclic_par_incor_complete}
  \item $\ihtriple{P}{C}{Q}$ is provable in \CPIL. \label{item:provable-cyclic-proofs-thm-cyclic_par_incor_complete}
 \end{enumerate}
\end{theorem}

To show \cref{thm:cyclic_par_incor_complete},
we define some concepts and show a lemma.

\begin{definition}[\CPIL proof with open leaves]
 A \emph{\CPIL -pre-proof with open leaves} is a pair $\mathcal{P} = (\mathcal{D,L})$, 
 where $\mathcal{D}$ is a finite derivation tree constructed according to the proof rules 
 in \cref{fig:rules-for-cyclic-proofs-of-par-incor}
 and $\mathcal{L}$ is a \emph{back-link partial function} 
 assigning to some open leaf of $\mathcal{D}$ a companion.
 For a \CPIL -pre-proof with open leaves $\mathprooffig{P} = (\mathprooffig{D}, L)$, 
 we call an open leaf which is not in the domain of $L$ a \emph{proper open leaf}.

 We define a \emph{\CPIL -proof with open leaves} as a pre-proof with open leaves
 satisfying the following \emph{global soundness condition}: 
 the rules except for \consILrule\ applied infinitely often 
 along every infinite path in the pre-proof with open leaves.
\end{definition}

We note that a \CPIL -proof with open leaves, where there is no proper open leaf, 
is a \CPIL -proof. 

\begin{lemma}
 \label[lemma]{lem:CPIL_cyclic_proof}
 If $\ihtriple{P}{C}{Q}$ is provable in \PIL\ then, 
 for any program $C'$ and assertion $R$, 
 there is a \CPIL -proof with open leaves of $\ihtriple{P}{C;C'}{R}$ such that 
 every proper open leaf is assigned $\ihtriple{Q}{C'}{R}$.
\end{lemma}

\begin{proof}
 We assume that $\ihtriple{P}{C}{Q}$ is provable in \PIL.
 We show the statement by induction on the proof of $\ihtriple{P}{C}{Q}$ in \PIL.
 We proceed by a case analysis on the last rule applied in the proof. 

 \proofcase{\axiomILrule} 
 In this case, ${C}\equiv{\mathemptyword}$.
 Then, ${P}\equiv{Q}$ holds and the proof of $\ihtriple{P}{C}{Q}$ is the following:
 \begin{center}
  \begin{inlineprooftree}
   \AxiomC{}
   \RightLabel{\axiomILrule}
   \UnaryInfC{$\htriple{Q}{\mathemptyword}{Q}$}
  \end{inlineprooftree}.
 \end{center}

 Noting that ${P}\equiv{Q}$, for arbitrary $C'$ and $R$,
 we have a \CPIL -proof with open leaves of $\htriple{P}{\mathemptyword ;C'}{R}$ 
 as follows:
 \begin{center}
  \begin{inlineprooftree}
   \AxiomC{$\htriple{Q}{C'}{R}$}
  \end{inlineprooftree}.
 \end{center}
 
 The only proper open leaf in this \CPIL -proof with open leaves is assigned
 $\htriple{Q}{C'}{R}$, as required.
 
 \proofcase{\assignPILrule}
 In this case, ${C}\equiv{\mathsubst{x}{E}}$.
 Then, ${P}\equiv{\mathsubstbox{Q}{\mathsubst{x}{E}}}$
 holds and the proof of $\ihtriple{P}{C}{Q}$ is the following:
 \begin{center}
  \begin{inlineprooftree}
   \AxiomC{\phantom{$\ihtriple{\mathsubstbox{P}{\mathsubst{x}{E}}}{\assign{x}{E}}{P}$}}
   \RightLabel{\assignPILrule}
   \UnaryInfC{$\ihtriple{\mathsubstbox{Q}{\mathsubst{x}{E}}}{\assign{x}{E}}{Q}$}
  \end{inlineprooftree}
 \end{center}

 For arbitrary $C'$ and $R$,
 we have a \CPIL -proof with open leaves of 
 $\htriple{\mathsubstbox{Q}{\mathsubst{x}{E}}}{{\mathsubst{x}{E}};C'}{R}$ 
 as follows:
 \begin{center}
  \begin{inlineprooftree}
   \AxiomC{$\ihtriple{Q}{C'}{R}$}
   \RightLabel{\assignCPILrule}
   \UnaryInfC{$\ihtriple{\mathsubstbox{Q}{\mathsubst{x}{E}}}{{\mathsubst{x}{E}};C'}{R}$}
  \end{inlineprooftree}.
 \end{center}

 The only proper open leaf in this \CPIL -proof with open leaves is assigned
 $\htriple{Q}{C'}{R}$, as required.

 \proofcase{\seqILrule}
 In this case, ${C}\equiv{C_{0};C_{1}}$.
 Then, the proof of $\ihtriple{P}{C}{Q}$ is the following:
 \begin{center}
  \begin{inlineprooftree}
   \AxiomC{}
   \DeduceC{$\ihtriple{P}{C_{0}}{R'}$}
   \AxiomC{}
   \DeduceC{$\ihtriple{R'}{C_{1}}{Q}$}
   \RightLabel{\seqILrule}
   \BinaryInfC{$\ihtriple{P}{C_{0};C_{1}}{Q}$}
  \end{inlineprooftree}.
 \end{center}
 Since $\ihtriple{P}{C}{R'}$ is provable in \PIL,
 induction hypothesis implies that, for any program $C''$ and assertion $R$, 
 there is a \CPIL -proof with open leaves of $\ihtriple{P}{C_{0};C''}{R}$ 
 such every open leaf is assigned $\ihtriple{R'}{C''}{R}$.
 Since $C''$ is arbitrary, 
 there is a \CPIL -proof with open leaves of $\ihtriple{P}{C_{0};C_{1};C'}{R}$
 such that every proper open leaf is assigned $\ihtriple{R'}{C_{1};C'}{R}$ 
 for any program $C'$.~(1)

 Then, since $\ihtriple{R'}{C_{1}}{Q}$ is provable in \PIL, 
 induction hypothesis implies that, for any program $C'$ and assertion $R$, 
 there is a \CPIL -proof with open leaves of $\ihtriple{R'}{C_{1};C'}{R}$ such that 
 every proper open leaf is assigned $\ihtriple{Q}{C'}{R}$.~(2)

 Putting (1) and (2) together gives us the \CPIL -proof with open leaves of 
 $\ihtriple{P}{C_{0}C_{1}C'}{R}$ such that 
 every proper open leaf is assigned $\ihtriple{Q}{C'}{R}$ as follows:
 \begin{center}
  \begin{inlineprooftree}
  \AxiomC{$\ihtriple{Q}{C'}{R}$}
   \RightLabel{(2)}
  \DeduceC{$\ihtriple{R'}{C_{1};C'}{R}$}
  \RightLabel{(1)}
  \DeduceC{$\ihtriple{P}{C_{0};C_{1}}{Q}$}
  \end{inlineprooftree}.
 \end{center}

 \proofcase{\consILrule}
 In this case, the proof of $\ihtriple{P}{C}{Q}$ is the following:
 \begin{center}
  \begin{inlineprooftree}
   \AxiomC{}
   \DeduceC{$\ihtriple{P'}{C}{Q'}$}
   \LeftLabel{(${P'}\models{P}$, ${Q}\models{Q'}$)}
   \RightLabel{\consILrule}
   \UnaryInfC{$\ihtriple{P}{C}{Q}$}
  \end{inlineprooftree}.
 \end{center}
 Since $\ihtriple{P'}{C}{Q'}$ is provable in \PIL,
 induction hypothesis implies that, for any program $C'$ and assertion $R$, 
 there is a \CPIL -proof with open leaves of $\ihtriple{P'}{C;C'}{R}$ 
 such every open leaf is assigned $\ihtriple{Q'}{C'}{R}$.

 Now, for arbitrary $C'$ and $R$,
 we have a \CPIL -proof with open leaves of 
 $\ihtriple{P}{C;C'}{R}$ as follows:
 \begin{center}
  \begin{inlineprooftree}
   \AxiomC{$\ihtriple{Q}{C'}{R}$}
   \LeftLabel{(${Q}\models{Q'}$)}
   \RightLabel{\consILrule}
   \UnaryInfC{$\ihtriple{Q'}{C'}{R}$}
   \RightLabel{(IH)}
   \DeduceC{$\ihtriple{P'}{C;C'}{R}$}
   \LeftLabel{(${P'}\models{P}$)}
   \RightLabel{\consILrule}
   \UnaryInfC{$\ihtriple{P}{C;C'}{R}$}
  \end{inlineprooftree}.
 \end{center}

 The proper open leaves in this \CPIL -proof with open leaves are assigned
 $\htriple{Q}{C'}{R}$, as required.

 \proofcase{\orPILrule}
 Assume that $\ihtriple{P}{\cOR{C_{0}}{C_{1}}}{Q}$ is provable in Hoare logic with 
 the following proof:
 \begin{center}
  \begin{inlineprooftree}
   \AxiomC{}
   \DeduceC{$\ihtriple{P}{C_{0}}{Q}$}
   \AxiomC{}
   \DeduceC{$\ihtriple{P}{C_{1}}{Q}$}
  \RightLabel{\orPILrule}
  \BinaryInfC{$\ihtriple{P}{\cOR{C_{0}}{C_{1}}}{Q}$}
\end{inlineprooftree}.
 \end{center}
 Since $\ihtriple{P}{C_{i}}{Q}$ is provable in \PIL\ for ${i}={0, 1}$,
 induction hypothesis implies that,
 for any program $C'$ and assertion $R$, 
 there is a \CPIL -proof with open leaves of $\ihtriple{P}{{C_{i}};{C'}}{R}$ 
 such every open leaf is assigned $\ihtriple{Q}{C'}{R}$ for ${i}={0, 1}$.
 We derive a \CPIL -proof with open leaves of $\ihtriple{Q}{C'}{R}$ as follows:
 \begin{center}
  \begin{inlineprooftree}
   \AxiomC{$\ihtriple{Q}{C'}{R}$}
   \RightLabel{(IH)}
   \DeduceC{$\ihtriple{P}{C_{0};C'}{Q}$}
   \AxiomC{$\ihtriple{Q}{C'}{R}$}
   \RightLabel{(IH)}
   \DeduceC{$\ihtriple{P}{C_{1};C'}{Q}$}
   \RightLabel{\orCPILrule}
   \BinaryInfC{$\ihtriple{P}{\cOR{C_{0}}{C_{1}};C'}{Q}$}
\end{inlineprooftree}.
 \end{center}

 \proofcase{\whilePILrule}
 In this case, ${C}\equiv{\while{B}{C_{0}}}$.
 Then, ${Q}\equiv{{\lnot B}\to{P}}$ holds and
 the proof of $\ihtriple{P}{C}{Q}$ is the following:
 \begin{center}
  \begin{inlineprooftree}
   \AxiomC{}
   \DeduceC{$\ihtriple{{B}\to{P}}{C_{0}}{P}$}
   \RightLabel{\whilePILrule}
   \UnaryInfC{$\ihtriple{P}{\while{B}{C_{0}}}{{\lnot B}\to{P}}$}
  \end{inlineprooftree}.
 \end{center} 

 Since $\ihtriple{{B}\to{P}}{C_{0}}{P}$ is provable in \PIL,
 the induction hypothesis implies that, for any program $C''$ and assertion $R$, 
 there is a \CPIL -proof with open leaves of $\ihtriple{{B}\to{P}}{C_{0};C''}{R}$
 such that every proper open leaf is assigned $\ihtriple{P}{C''}{R}$.
 Since $C''$ is arbitrary,
 for any program $C'$ and assertion $R$, 
 there is a \CPIL -proof with open leaves of 
 $\ihtriple{{B}\to{P}}{C_{0};{\while{B}{C_{0}}};C'}{R}$
 such that every proper open leaf is assigned $\ihtriple{P}{{\while{B}{C_{0}}};C'}{R}$.

 We derive a \CPIL -proof with open leaves of 
 $\htriple{P}{{\while{B}{C_{0}}};C'}{R}$ as follows:
 \begin{center}
  \begin{inlineprooftree}
   \AxiomC{$\ihtriple{{\lnot B}\to{P}}{C'}{R}$}
   
   \AxiomC{$\ihtriple{P}{{\while{B}{C_{0}}};C'}{R}$\tikzmark{bud5}}
   \RightLabel{(IH)}
   \DeduceC{$\ihtriple{{B}\to{P}}{C;{{\while{B}{C_{0}}};C'}}{R}$}
   \RightLabel{\whileCPILrule\tikzmark{mid5}}
   \BinaryInfC{$\ihtriple{P}{{\while{B}{C_{0}}};C'}{R}$ \tikzmark{comp5}}
  \end{inlineprooftree}.
   \begin{tikzpicture}[remember picture, overlay, thick, relative, auto, rounded corners, line width=1.0pt]
    \coordinate (b) at ({pic cs:bud5});
    \coordinate (root) at ({pic cs:comp5});
    \coordinate (M) at ({pic cs:mid5});
    
    \draw[->] let \p1=($(M)-(b)$), \p2=($(root)-(b)$) in
    (b) -- ++(\x1 + 0.25em,0) -- ++(0,\y2) -- (root);
   \end{tikzpicture}
 \end{center}
 In the \CPIL -proof with open leaves above,
 any occurrence of $\ihtriple{P}{{\while{B}{C_{0}}};C'}{R}$ as a leaf has a back-link 
 to the root.
 Then, each proper open leaf in the \CPIL -proof with open leaves above is
 assigned $\ihtriple{{\lnot B}\to{P}}{C'}{R}$.
\end{proof}

Now, we show the completeness theorem.

\begin{proof}[Proof of \cref{thm:cyclic_par_incor_complete}]
 Fix $\ihtriple{P}{C}{Q}$ be a partial reverse Hoare triple.

 \noindent
 \cref{item:validness-thm-cyclic_par_incor_complete}$\Rightarrow$\cref{item:provable-reverse_Hoare-thm-cyclic_par_incor_complete}:
 By \cref{thm:rev_hoare_complete}.

 \noindent
 \cref{item:provable-reverse_Hoare-thm-cyclic_par_incor_complete}$\Rightarrow$\cref{item:provable-cyclic-proofs-thm-cyclic_par_incor_complete}:
 Assume $\ihtriple{P}{C}{Q}$ is provable in \PIL.
 By \cref{lem:CPIL_cyclic_proof}, 
 there is a \CPIL -proof with open leaves of $\ihtriple{P}{C}{Q}$ such that 
 every proper open leaf is assigned $\ihtriple{Q}{\mathemptyword}{Q}$.
 Since $\ihtriple{Q}{\mathemptyword}{Q}$ can be the conclusion of \axiomILrule,
 there is a \CPIL -proof of $\ihtriple{P}{C}{Q}$.

 \noindent
 \cref{item:provable-cyclic-proofs-thm-cyclic_par_incor_complete}$\Rightarrow$\cref{item:validness-thm-cyclic_par_incor_complete}:
 By \cref{thm:par-incorrect_soundness}.
\end{proof}
\section{Conclusion}
\label{sec:conc}
We have given ordinary and cyclic proof systems for partial reverse Hoare logic.
Then, we have shown their soundness and relative completeness.
Although the semantics of partial reverse Hoare logic is the dual of ``total'' Hoare logic,
assertions in the rule for the while loop are the dual of these in ``partial'' Hoare logic.
Comparing cyclic proofs with ordinary proofs, we do not need to find loop invariants.
This is an advantage of cyclic proofs for proof search.

We wonder whether $\mathsetweakestpresy$-expressiveness is necessary 
for relative completeness. 
J.~A.~Bergstra and J.~V.~Tucker \cite{Bergstra1982} showed that
the expressiveness of the language of assertions, 
which means that the language can express the weakest liberal pre-conditions for
any assertion and any program, is not necessary 
for the relative completeness of partial Hoare logic.
We conjecture that a similar result holds in partial reverse Hoare logic.

Other future work would be 
(1) to extend partial reverse Hoare logic, for example, by separation logic,
(2) to define cyclic proof systems for other Hoare-style logics,
 and
(3) to study a method to find loop invariants from cyclic proofs.

\ifacm
\begin{acks}
\else
\section*{Acknowledgements}
\fi 
We would like to thank Quang Loc Le, James Brotherston, Koji Nakazawa, Daisuke Kimura, and Tatsuya Abe for their valuable comments.

\ifacm
\end{acks}
\fi

\ifacm
\bibliographystyle{ACM-Reference-Format}
\else
\bibliographystyle{plain}
\fi
\bibliography{refs}

\appendix
\section{Construction of weakest pre-condition assertion}
\label[appendix]{app:wpr}
We construct a weakest pre-condition assertion $\mathweakestpre{C}{Q}$.
In this appendix, we assume that our language includes 
addition operator $+$, subtraction operator $-$, multiplication operator $*$,
division operator $/$, and reminder operator $\%$.

We abbreviate ${x}={{a}\%{(1+(1+i)*b)}}$ to $\mathof{\beta}{a, b, i, x}$.
It is so-called \emph{G\"{o}del's predicate $\beta$} \cite{Winskel1993}.
For G\"{o}del's predicate $\beta$, the following statement holds.

\begin{fact}
 \label[fact]{fact:Godel-beta}
 For any finite sequence of natural numbers $a_{0}, \dots, a_{k}$, and
 any natural number $j$ (${0}\leq j \leq k$),
 there exists two natural numbers $n$ and $m$ such that
 ${x}={a_{j}}$ holds if and only if $\mathof{\beta}{n, m, j, x}$ holds.
\end{fact}

The above fact means that any finite sequence of natural numbers can be encoded
as two natural numbers $n$ and $m$.


\begin{figure*}[t]
  \begin{align*}
  {\mathweakestpre{\mathemptyword}{Q}}&\equiv{Q}  \\
  {\mathweakestpre{\assign{x}{E}}{Q}}&\equiv{\mathsubstbox{Q}{\mathsubst{x}{E}}}   \\
  {\mathweakestpre{C_{1};C_{2}}{Q}}&\equiv{\mathweakestpre{C_{1}}{\mathweakestpre{C_{2}}{Q}}} \\
   {\mathweakestpre{\cOR{C_{1}}{C_{2}}}{Q}}&\equiv{{\mathweakestpre{C_{1}}{Q}}\lor{\mathweakestpre{C_{2}}{Q}}} 
 \end{align*}
 \begin{multline*}
  {\mathweakestpre{\while{B}{C}}{Q}}\equiv 
   \exists k \exists m \exists n \forall y_{1} \dots \forall y_{l} \forall y'_{1} \dots \forall y'_{l} \forall y''_{1} \dots \forall y''_{l} \forall y'''_{1} \dots \forall y'''_{l}  \\
  ({\mathof{F_{l}}{n,m}} \land 
  {\mathof{S_{l}}{k, m, n, y_{1}, \dots, y_{l}, y'_{1}, \dots, y'_{l}, y''_{1}, \dots, y''_{l}, B, C}} \land 
  {\mathof{T_{l}}{k, m, n, y'''_{1}, \dots, y'''_{l}, B, Q}} ), 
 \end{multline*}
  where ${{\mathFVof{P}}\cup{\mathVarof{\while{B}{C}}}}={\mathsetextension{x_{1}, \dots, x_{l}}}$,
  \[
  {\mathof{F_{l}}{n, m}}\equiv{\mleft(\mathof{\beta}{n, m, 0, x_{1}} \land \dots \land \mathof{\beta}{n, m, l-1, x_{l}}\mright)},
  \]
  \begin{multline*}
  {\mathof{S_{l}}{k, m, n, y'_{1}, \dots, y'_{l}, y''_{1}, \dots, y''_{l}, y'''_{1}, \dots, y'''_{l}, B, C}}\equiv  
   (0 < k \to %
   (%
   \forall i ( %
   (0 \leq i \land i < k) \to \\ %
   (\mathof{\beta}{n, m, l*i, y_{1}} \land \dots \land \mathof{\beta}{n, m, l*(i+1)-1, y_{l}})%
  \land \\%
 ( \mathof{\beta}{n, m, l*(i+1), y'_{1}} \land \dots \land \mathof{\beta}{n, m, l*(i+2)-1, y'_{l}}) \\ 
  \to 
  ( {\mathsubstbox{B}{{\mathsubst{x_{1}}{y_{1}}}, \dots, {\mathsubst{x_{l}}{y_{l}}}}}
  \land \\
   \mathsubstbox{
   ({\mathweakestpre{C}{x_{1}=y'_{1} \land \dots \land x_{l}=y'_{l}}} %
   \to (x_{1}=y_{1} \land \dots \land x_{l}=y_{l})
  )}{{\mathsubst{x_{1}}{y''_{1}}}, \dots, {\mathsubst{x_{l}}{y''_{l}}}}
  )%
  ) %
  )
   ),
 \end{multline*}
 and
 \begin{multline*}
 {\mathof{T_{l}}{k, m, n, y'''_{1}, \dots, y'''_{l}, B, Q}}\equiv
  (\mathof{\beta}{n, m, l*k, y'''_{1}} \land \dots \land \mathof{\beta}{n, m, l*(k+1)-1, y'''_{l}} \to \\
  ((\lnot \mathsubstbox{B}{{\mathsubst{x_{1}}{y'''_{1}}}, \dots, {\mathsubst{x_{l}}{y'''_{l}}}}) 
  \land \\
 (\mathsubstbox{Q}{{\mathsubst{x_{1}}{y'''_{1}}}, \dots, {\mathsubst{x_{l}}{y'''_{l}}}})
) )  
 \end{multline*}

 \caption{Weakest pre-condition}
 \label{fig:weakest-precondition}
\end{figure*}

\begin{definition}[Weakest pre-condition assertion]
 \label[definition]{definition:weakest-pre-condition-assertion}
 For an assertion $P$ and a program $C$,
 we inductively define an assertion $\mathweakestpre{C}{Q}$ 
 in \cref{fig:weakest-precondition}.
\end{definition}

\begin{proposition}
 \label[proposition]{lem:eq-between-weakestpre-and-setweakestpre}
 ${\sigma}\in{\mathsetweakestpre{C}{Q}}$
 if and only if
 ${\sigma}\models{\mathweakestpre{C}{Q}}$
 for any state $\sigma$.
\end{proposition}

\begin{proof}
 We show the statement by induction on construction of $C$.
 We consider cases according to the form of $C$.

 \proofcase{(${C}\equiv{\mathemptyword}$)} 
 In this case, ${\mathweakestpre{C}{Q}}\equiv{Q}$.
 
 If ${\sigma}\in{\mathsetweakestpre{C}{Q}}$, 
 then we have ${\sigma}\models{Q}\equiv{\mathweakestpre{C}{Q}}$.

 Assume ${\sigma}\models{\mathweakestpre{C}{Q}}\equiv{Q}$.
 Then, we have
 ${\config{\mathemptyword}{\sigma} \evaluation^{*} \config{\mathemptyword}{\sigma}}$ and
 ${{\sigma} \models {Q}}$.
 We have ${\sigma}\in{\mathsetweakestpre{C}{Q}}$.

 \proofcase{(${C}\equiv{\assign{x}{E}}$)} 
 In this case, ${\mathweakestpre{C}{Q}}\equiv{\mathsubstbox{Q}{\mathsubst{x}{E}}}$.

 Assume ${\sigma}\in{\mathsetweakestpre{C}{Q}}$.
 Then, we have
 ${\config{\assign{x}{E}}{\sigma} \evaluation \config{\mathemptyword}{\mathsubstbox{\sigma}{\mathsubst{x}{\sem{E}\sigma}}}}$ and
 ${{\mathsubstbox{\sigma}{\mathsubst{x}{\sem{E}\sigma}}} \models {Q}}$.
 By \cref{lem:assert_subst},
 ${\sigma} \models {\mathsubstbox{Q}{\mathsubst{x}{E}}}$ holds.
 Hence, we have ${\sigma}\models{\mathweakestpre{C}{Q}}$.

 Assume ${\sigma}\models{\mathweakestpre{C}{Q}}\equiv{\mathsubstbox{Q}{\mathsubst{x}{E}}}$.
 Then, we have 
 ${\config{\assign{x}{E}}{\sigma} \evaluation \config{\mathemptyword}{\mathsubstbox{\sigma}{\mathsubst{x}{\sem{E}\sigma}}}}$ and
 ${{\mathsubstbox{\sigma}{\mathsubst{x}{\sem{E}\sigma}}} \models {Q}}$.
 Thus, ${\sigma}\in{\mathsetweakestpre{C}{Q}}$ holds.

 \proofcase{(${C}\equiv{C_{1}; C_{2}}$)} 
 In this case, ${\mathweakestpre{C}{Q}}\equiv{\mathweakestpre{C_{1}}{\mathweakestpre{C_{2}}{Q}}}$.

 Assume ${\sigma}\in{\mathsetweakestpre{C}{Q}}$.
 Then, there exists a state $\sigma'$ such that
 ${\config{C_{1}; C_{2}}{\sigma} \evaluation^{*} \config{\mathemptyword}{\sigma'}}$ and
 ${{\sigma'} \models {Q}}$ hold.
 By \cref{lem:property-config} \cref{item:sequencing-lem-property-config},
 there exists $\sigma''$ such that 
 ${\config{C_{1}}{\sigma}} \evaluation^{*} {\config{\mathemptyword}{\sigma''}}$ and 
 ${\config{C_{2}}{\sigma''}} \evaluation^{*} {\config{\mathemptyword}{\sigma'}}$ hold.
 Since ${\config{C_{2}}{\sigma''}} \evaluation^{*} {\config{\mathemptyword}{\sigma'}}$ and
 ${{\sigma'} \models {Q}}$ hold,
 we have ${\sigma''}\in{\mathsetweakestpre{C_{2}}{Q}}$.
 By induction hypothesis, we have ${\sigma''}\models{\mathweakestpre{C_{2}}{Q}}$.
 Because of ${\config{C_{1}}{\sigma}} \evaluation^{*} {\config{\mathemptyword}{\sigma''}}$,
 we see ${\sigma}\in{\mathsetweakestpre{C_{1}}{\mathweakestpre{C_{2}}{Q}}}$.
 By induction hypothesis,
 we have ${\sigma}\models{\mathweakestpre{C_{1}}{\mathweakestpre{C_{2}}{Q}}}$.

 Assume ${\sigma}\models{\mathweakestpre{C}{Q}}$.
 Then, we have ${\sigma}\models{\mathweakestpre{C_{1}}{\mathweakestpre{C_{2}}{Q}}}$.
 By induction hypothesis,
 we have ${\sigma}\in{\mathsetweakestpre{C_{1}}{\mathweakestpre{C_{2}}{Q}}}$.
 Hence, there exists $\sigma''$ such that 
 ${\config{C_{1}}{\sigma}} \evaluation^{*} {\config{\mathemptyword}{\sigma''}}$ and 
 ${\sigma''}\models{\mathweakestpre{C_{2}}{Q}}$ hold.
 By induction hypothesis, we have ${\sigma''}\in{\mathsetweakestpre{C_{2}}{Q}}$.
 Therefore, there exists a state $\sigma'$ such that
 ${\config{C_{2}}{\sigma''}} \evaluation^{*} {\config{\mathemptyword}{\sigma'}}$ and
 ${{\sigma'} \models {Q}}$.
 By \cref{lem:property-config} \cref{item:sequencing-lem-property-config},
 ${\config{C_{1}; C_{2}}{\sigma} \evaluation^{*} \config{\mathemptyword}{\sigma'}}$ holds.
 Thus, ${\sigma}\in{\mathsetweakestpre{C}{Q}}$ holds.

 \proofcase{(${C}\equiv{\cOR{C_{1}}{C_{2}}}$)} 
  In this case, ${\mathweakestpre{C}{Q}}\equiv{{\mathweakestpre{C_{1}}{Q}}\lor{\mathweakestpre{C_{2}}{Q}}}$.

 Assume ${\sigma}\in{\mathsetweakestpre{C}{Q}}$.
 Then, there exists a state $\sigma'$ such that
 ${\config{\cOR{C_{1}}{C_{2}}}{\sigma} \evaluation^{*} \config{\mathemptyword}{\sigma'}}$ and
 ${{\sigma'} \models {Q}}$ hold.
 We see either 
 $\config{\cOR{C_{0}}{C_{1}}}{\sigma} \evaluation \config{C_{0}}{\sigma} \evaluation^{*} \config{\mathemptyword}{\sigma'}$
 or 
 $\config{\cOR{C_{0}}{C_{1}}}{\sigma} \evaluation \config{C_{1}}{\sigma} \evaluation^{*} \config{\mathemptyword}{\sigma'}$ 
 holds.
 Then, we have either ${\sigma}\in{\mathsetweakestpre{C_{0}}{Q}}$ or 
 ${\sigma}\models{\mathsetweakestpre{C_{1}}{Q}}$.
 By induction hypothesis, we have either ${\sigma}\models{\mathweakestpre{C_{0}}{Q}}$ or 
 ${\sigma}\models{\mathweakestpre{C_{1}}{Q}}$.
 Thus, ${\sigma}\models{{\mathweakestpre{C_{0}}{Q}}\lor{\mathweakestpre{C_{1}}{Q}}}\equiv{\mathweakestpre{C}{Q}}$ holds.

 Assume ${\sigma}\models{\mathweakestpre{C}{Q}}$.
 Then, ${\sigma}\models{{\mathweakestpre{C_{0}}{Q}}\lor{\mathweakestpre{C_{1}}{Q}}}$.
 We have ${\sigma}\models{\mathweakestpre{C_{0}}{Q}}$ or 
 ${\sigma}\models{\mathweakestpre{C_{1}}{Q}}$ holds.
 By induction hypothesis,
 we have either ${\sigma}\in{\mathsetweakestpre{C_{0}}{Q}}$ or 
 ${\sigma}\models{\mathsetweakestpre{C_{1}}{Q}}$.

 Assume ${\sigma}\in{\mathsetweakestpre{C_{0}}{Q}}$.
 Then, there exists a state $\sigma'$ such that
 $\config{C_{0}}{\sigma} \evaluation^{*} \config{\mathemptyword}{\sigma'}$ and
 ${{\sigma'} \models {Q}}$ hold.
 Because $\config{\cOR{C_{0}}{C_{1}}}{\sigma} \evaluation \config{C_{0}}{\sigma}$ holds,
 we have $\config{\cOR{C_{0}}{C_{1}}}{\sigma} \evaluation^{*} \config{\mathemptyword}{\sigma'}$.
 Hence, we have ${\sigma}\in{\mathsetweakestpre{\cOR{C_{0}}{C_{1}}}{Q}}$.

 In the similar way, we have ${\sigma}\in{\mathsetweakestpre{\cOR{C_{0}}{C_{1}}}{Q}}$
 if ${\sigma}\in{\mathsetweakestpre{C_{1}}{Q}}$ holds.

 Thus, ${\sigma}\in{\mathsetweakestpre{C}{Q}}$ holds.

 \proofcase{(${C}\equiv{\while{B}{C_{0}}}$)} 
 Let ${{\mathFVof{P}}\cup{\mathVarof{\while{B}{C}}}}={\mathsetextension{x_{1}, \dots, x_{l}}}$.

 Assume ${\sigma}\in{\mathsetweakestpre{C}{Q}}$.
 Then, there exists a state $\sigma'$ such that 
 ${\config{\while{B}{C_{0}}}{\sigma} \evaluation^{*} \config{\mathemptyword}{\sigma'}}$ and
 ${{\sigma'} \models {Q}}$ hold.
 By \cref{lem:property-config} \cref{item:while-lem-property-config},
 there exist states $\sigma_{0}, \dots, \sigma_{k}$ 
 such that ${\sigma_{0}}\equiv{\sigma}$, ${\sigma_{k}}\equiv{\sigma'}$, 
 ${\sigma_{k}}\models{\lnot B}$ hold, and
 $k>0$ implies that
 ${\config{C}{\sigma_{i}}} \evaluation^{*} {\config{\mathemptyword}{\sigma_{i+1}}}$ and
 ${\sigma_{i}}\models{B}$
 for each $i=0, \dots, k-1$.
 Let ${s_{(i*l)+(j-1)}}={\mathof{\sigma_{i}}{x_{j}}}$ 
 for each $i=0, \dots, k$ and each $j=1, \dots, l$.
  By \cref{fact:Godel-beta},
 $\mathsequence{s_{h}}{{0} \leq {h} \leq {k*(l+1)-1}}$
 can be encoded as two natural numbers $n$ and $m$.
 Then,
 \[
  {\sigma}\models{\mleft(\mathof{\beta}{n, m, 0, x_{1}} \land \dots \land \mathof{\beta}{n, m, l-1, x_{l}}\mright)}
 \]
 holds.
 
 Assume $k>0$. For each $i=0, \dots, k-1$,
 \[
  {}\models  (\mathof{\beta}{n, m, l*i, y_{1}} \land \dots \land \mathof{\beta}{n, m, l*(i+1)-1, y_{l}}) \to  ({\mathsubstbox{B}{{\mathsubst{x_{1}}{y_{1}}}, \dots, {\mathsubst{x_{l}}{y_{l}}}}})
 \]
 and
 \begin{multline*}
  {}\models
  (\mathof{\beta}{n, m, l*i, y_{1}} \land \dots \land \mathof{\beta}{n, m, l*(i+1)-1, y_{l}}) \land \\
  ( \mathof{\beta}{n, m, l*(i+1), y'_{1}} \land \dots \land \mathof{\beta}{n, m, l*(i+2)-1, y'_{l}})
  \to \\
  ({\mathweakestpre{C}{x_{1}=y'_{1} \land \dots \land x_{l}=y'_{l}}} %
  \to \\
  \mathsubstbox{
   (x_{1}=y_{1} \land \dots \land x_{l}=y_{l})
  )}{{\mathsubst{x_{1}}{y''_{1}}}, \dots, {\mathsubst{x_{l}}{y''_{l}}}}
 \end{multline*}.
 Hence, 
 \[
 {}
 \models 
 \mathof{S_{l}}{k, m, n, y'_{1}, \dots, y'_{l}, y''_{1}, \dots, y''_{l}, y'''_{1}, \dots, y'''_{l}, B, C}
 \]
 By ${\sigma_{k}}\models{\lnot B}$ and ${\sigma_{k}}\models{Q}$, we have
 \begin{multline*}
  {}\models
  (\mathof{\beta}{n, m, l*k, y'''_{1}} \land \dots \land \mathof{\beta}{n, m, l*(k+1)-1, y'''_{l}} \to \\
  ((\lnot \mathsubstbox{B}{{\mathsubst{x_{1}}{y'''_{1}}}, \dots, {\mathsubst{x_{l}}{y'''_{l}}}}) 
  \land \\
 (\mathsubstbox{Q}{{\mathsubst{x_{1}}{y'''_{1}}}, \dots, {\mathsubst{x_{l}}{y'''_{l}}}})
)).
 \end{multline*}
 Thus,
 \begin{multline*}
 {\sigma}\models \exists k \exists m \exists n \forall y_{1} \dots \forall y_{l} \forall y'_{1} \dots \forall y'_{l} \forall y''_{1} \dots \forall y''_{l} \forall y'''_{1} \dots \forall y'''_{l}  \\
  ({\mathof{F_{l}}{n,m}} \land 
  {\mathof{S_{l}}{k, m, n, y_{1}, \dots, y_{l}, y'_{1}, \dots, y'_{l}, y''_{1}, \dots, y''_{l}, B, C}} \land \\
  {\mathof{T_{l}}{k, m, n, y'''_{1}, \dots, y'''_{l}, B, Q}} )
 \end{multline*}

 Assume ${\sigma}\models{\mathweakestpre{C}{Q}}$.
 Then, there exist natural numbers $\bar{k}$, $\bar{m}$ and $\bar{n}$ such that
 \begin{align}
  \mathsubstbox{\sigma}{
  \begin{aligned}
   &\mathsubst{k}{\bar{k}}, \mathsubst{m}{\bar{m}}, \mathsubst{n}{\bar{n}}, \\
   &{\mathsubst{\mathvect{y}}{\mathvect{c}}}, {\mathsubst{\mathvect{y'}}{\mathvect{c'}}}, 
{\mathsubst{\mathvect{y''}}{\mathvect{c''}}}, {\mathsubst{\mathvect{y'''}}{\mathvect{c'''}}}
  \end{aligned}
  }  &\models  {\mathof{F_{l}}{n,m}}, \label{eq:first-valid} \\
 \mathsubstbox{\sigma}{
 \begin{aligned}
   &\mathsubst{k}{\bar{k}}, \mathsubst{m}{\bar{m}}, \mathsubst{n}{\bar{n}}, \\ 
   &{\mathsubst{\mathvect{y}}{\mathvect{c}}}, {\mathsubst{\mathvect{y'}}{\mathvect{c'}}}, 
   {\mathsubst{\mathvect{y''}}{\mathvect{c''}}}, {\mathsubst{\mathvect{y'''}}{\mathvect{c'''}}}
 \end{aligned} 
  }  
  &\models 
  {\mathof{S_{l}}{k, m, n, y_{1}, \dots, y_{l}, y'_{1}, \dots, y'_{l}, y''_{1}, \dots, y''_{l}, B, C_{0}}}, \label{eq:second-valid}   \\
 \intertext{and}
  \mathsubstbox{\sigma}{
  \begin{aligned}
   &\mathsubst{k}{\bar{k}}, \mathsubst{m}{\bar{m}}, \mathsubst{n}{\bar{n}}, \\ 
   &{\mathsubst{\mathvect{y}}{\mathvect{c}}}, {\mathsubst{\mathvect{y'}}{\mathvect{c'}}}, 
   {\mathsubst{\mathvect{y''}}{\mathvect{c''}}}, {\mathsubst{\mathvect{y'''}}{\mathvect{c'''}}}
  \end{aligned} 
  } 
  &\models 
  {\mathof{T_{l}}{k, m, n, y'''_{1}, \dots, y'''_{l}, B, Q}}, \label{eq:third-valid} 
 \end{align}
 where 
 ${\mathsubst{\mathvect{y}}{\mathvect{c}}}\equiv{\mathsubst{y_{1}}{c_{1}}, \dots, \mathsubst{y_{l}}{c_{l}}}$,
 ${\mathsubst{\mathvect{y'}}{\mathvect{c'}}}\equiv{\mathsubst{y'_{1}}{c'_{1}}, \dots, \mathsubst{y'_{l}}{c'_{l}}}$,
 ${\mathsubst{\mathvect{y''}}{\mathvect{c''}}}\equiv{\mathsubst{y''_{1}}{c''_{1}}, \dots, \mathsubst{y''_{l}}{c''_{l}}}$, and
 ${\mathsubst{\mathvect{y'''}}{\mathvect{c'''}}}\equiv{\mathsubst{y'''_{1}}{c'''_{1}}, \dots, \mathsubst{y'''_{l}}{c'''_{l}}}$,
 for all natural numbers
 $c_{1}, \dots, c_{l}, c'_{1}, \dots, c'_{l}, c''_{1}, \dots, c''_{l}, c'''_{1}, \dots, c'''_{l}$.
 By \cref{fact:Godel-beta},
 two natural numbers $n$ and $m$ encode
 some finite sequence of natural numbers $\mathsequence{s_{h}}{{0} \leq {h} \leq }$.
 By \cref{eq:first-valid} and ${\mathFVof{P}}={\mathsetextension{x_{1}, \dots, x_{l}}}$,
 we have 
 \[
 {\sigma}\models{\mleft(\mathof{\beta}{n, m, 0, x_{1}} \land \dots \land \mathof{\beta}{n, m, l-1, x_{l}}\mright)}.
 \]
 By \cref{eq:second-valid},
  \begin{multline}
  {\sigma}\models    
   \mathsubstbox{\mathweakestpre{C_{0}}{x_{1}=s_{i*l} \land \dots \land x_{l}=s_{i*(l+1)-1}}}{\mathsubst{\mathvect{x}}{\mathvect{c''}}} \\
   \to 
   (c''_{1}=s_{i*(l+1)} \land \dots \land c''_{l}=s_{i*(l+2)-1}),
   \label{eq:step-case-instance}
 \end{multline}
 where ${\mathsubst{\mathvect{x}}{\mathvect{c''}}}\equiv{\mathsubst{x_{1}}{c''_{1}}, \dots, \mathsubst{x_{l}}{c''_{l}}}$,
 for each $i=0, \dots, \bar{k}-1$.
 By \cref{eq:third-valid}, we have
\[
 {\sigma}\models 
 \lnot \mathsubstbox{B}{{\mathsubst{x_{1}}{s_{l*\bar{k}}}}, \dots, {\mathsubst{x_{l}}{s_{l*(\bar{k}+1)-1}}}}
 \land 
 \mathsubstbox{Q}{{\mathsubst{x_{1}}{s_{l*\bar{k}}}}, \dots, {\mathsubst{x_{l}}{s_{l*(\bar{k}+1)-1}}}}.   
 \]
 Let $\sigma_{i}$ be a state for $i=0, \dots, \bar{k}$, where
 ${\mathof{\sigma_{i}}{x_{j}}}={s_{i*l+j-1}}$.
 Then, ${\sigma_{0}}\equiv{\sigma}$, ${\sigma_{k}}\equiv{\sigma'}$, and
 ${\sigma_{k}}\models{\lnot B}$ hold.
 By \cref{eq:step-case-instance},
 we have
 ${\config{C_{0}}{\sigma_{i}}} \evaluation^{*} {\config{\mathemptyword}{\sigma_{i+1}}}$, and
 ${\sigma_{i}}\models{B}$
 hold for each $i=0, \dots, \bar{k}-1$.
 By \cref{lem:property-config} \cref{item:while-lem-property-config},
 $\config{\while{B}{C_{0}}}{\sigma} \evaluation^{*} \config{\mathemptyword}{\sigma_{k}}$ hold.
 Because of ${\sigma_{k}}\models{Q}$, we have ${\sigma}\in{\mathsetweakestpre{C}{Q}}$.
\end{proof}

\end{document}